\documentclass[reqno,11pt]{amsart}
\usepackage[utf8]{inputenc}
\usepackage{enumitem}
\usepackage{graphicx}
\usepackage{amscd}
\usepackage{slashed}
\usepackage{amssymb}
\usepackage{mathtools} 
\usepackage[off]{auto-pst-pdf} 
\usepackage[mathscr]{eucal}
\numberwithin{equation}{section}
\allowdisplaybreaks[1]

\newcommand{\bigslant}[2]{{\raisebox{.3em}{$#1$}\big/ \raisebox{-.3em}{$#2$} }}

\title[Linear Bosonic Quantum Fields for Causal Variational Principles]{Linear Bosonic Quantum Field Theories \\
Arising from Causal Variational Principles}

\author[C.\ Dappiaggi]{Claudio Dappiaggi}
\address{Dipartimento di Fisica \\ Universit{\`a} degli Studi di Pavia\\ and INFN, Sezione di Pavia \\ Via Bassi, 6 --  I-27100 Pavia \\ Italy}
\email{claudio.dappiaggi@unipv.it}

\author[F.\ Finster]{Felix Finster}
\address{Fakult\"at f\"ur Mathematik \\ Universit\"at Regensburg \\ D-93040 Regensburg \\ Germany}
\email{finster@ur.de, marco.oppio.r@gmail.com} 

\author[M.\ Oppio]{Marco Oppio \\ \\ December 2021}

\newtheorem{Def}{Definition}[section]
\newtheorem{Thm}[Def]{Theorem}
\newtheorem{Prp}[Def]{Proposition}
\newtheorem{Lemma}[Def]{Lemma}
\newtheorem{Remark}[Def]{Remark}

\newtheorem{Assumption}[Def]{Assumption}
\newcommand{\Thanks}{\vspace*{.5em} \noindent \thanks}
\newcommand{\beq}{\begin{equation}}
\newcommand{\eeq}{\end{equation}}
\newcommand{\Proof}{\begin{proof}}
\newcommand{\QED}{\end{proof} \noindent}

\newcommand{\la}{\langle}
\newcommand{\ra}{\rangle}

\newcommand{\Sl}{\mathopen{\prec}}
\newcommand{\Sr}{\mathclose{\succ}}
\newcommand{\C}{\mathbb{C}}
\newcommand{\R}{\mathbb{R}}
\newcommand{\1}{\mbox{\rm 1 \hspace{-1.05 em} 1}}

\newcommand{\N}{\mathbb{N}}
\renewcommand{\H}{\mathscr{H}}

\newcommand{\lla}{\langle\!\langle}
\newcommand{\rra}{\rangle\!\rangle}

\newcommand{\F}{{\mathscr{F}}}
\newcommand{\Dir}{{\mathcal{D}}}

\renewcommand{\L}{{\mathcal{L}}}
\newcommand{\Sact}{{\mathcal{S}}}
\newcommand{\s}{{\mathfrak{s}}}

\newcommand{\loc}{\text{\rm{loc}}}
\renewcommand{\sc}{\text{\rm{sc}}}

\newcommand{\scrM}{\myscr M}

\newcommand{\J}{\mathfrak{J}}
\newcommand{\Jlin}{\mathfrak{J}^\text{\rm{\tiny{lin}}}}
\newcommand{\Jtest}{\mathfrak{J}^\text{\rm{\tiny{test}}}}
\newcommand{\Jvary}{\mathfrak{J}^\text{\rm{\tiny{vary}}}}
\newcommand{\Ctest}{C^\text{\rm{\tiny{test}}}}
\newcommand{\Gdiff}{\Gamma^\text{\rm{\tiny{diff}}}}

\newcommand{\Gtest}{\Gamma^\text{\rm{\tiny{test}}}}
\newcommand{\Jdiff}{\mathfrak{J}^\text{\rm{\tiny{diff}}}}

\newcommand{\bitem}{\begin{itemize}[leftmargin=2em]}
\newcommand{\eitem}{\end{itemize}}
\newcommand{\itemD}{\item[{\raisebox{0.125em}{\tiny $\blacktriangleright$}}]}

\newcommand{\tb}{|\!|\!|} 
\newcommand{\hol}{\text{\rm{hol}}}
\newcommand{\ah}{\text{\rm{ah}}}
\newcommand{\B}{{\mathscr{B}}}

\DeclareFontFamily{OT1}{rsfso}{}
\DeclareFontShape{OT1}{rsfso}{m}{n}{ <-7> rsfso5 <7-10> rsfso7 <10-> rsfso10}{}
\DeclareMathAlphabet{\myscr}{OT1}{rsfso}{m}{n}

\DeclareMathOperator{\im}{Im}

\DeclareMathOperator{\supp}{supp}

\renewcommand{\u}{\mathfrak{u}}
\renewcommand{\v}{\mathfrak{v}}
\newcommand{\w}{\mathfrak{w}}

\newcommand{\h}{\mathfrak{h}}

\begin{document}

\maketitle

\begin{abstract}
It is shown that the linearized fields of causal variational principles
give rise to linear bosonic quantum field theories.
The properties of these field theories are studied and compared
with the axioms of local quantum physics.
Distinguished quasi-free states are constructed.
\end{abstract} 

\tableofcontents

\section{Introduction}
The theory of {\em{causal fermion systems}} is a novel approach to fundamental physics
(see the reviews~\cite{dice2014, review, dice2018}, the textbook~\cite{cfs} or the website~\cite{cfsweblink}). Recently, the connection to quantum states has been
established~\cite{fockbosonic, fockfermionic}. Moreover, in~\cite{linhyp} the Cauchy problem
for {\em{linearized fields}} has been studied in the more general setting of {\em{causal variational principles}}
with energy methods inspired by the theory of hyperbolic partial differential equations.
Based on these concepts and results, in the present paper we show that causal variational principles
give rise to a class of linear bosonic quantum field theories generated by the
linearized fields.
We formulate these theories in the familiar language of axiomatic quantum field theory
in terms of an {\em{algebra of fields}} satisfying canonical commutation relations.
This makes it possible to verify whether and to specify in which sense
our quantum field theories satisfy the axioms of local quantum physics.
We find that the axiom of microlocality and the time slice axiom
are violated in the strict sense. But they still hold on macroscopic scales in the following sense.
In order to formulate the {\em{time slice property}}, instead of working with field operators
on a Cauchy surface, one must consider a {\em{Cauchy surface layer}}, which can be thought of
as a Cauchy surface times a small time interval.
The axiom of {\em{microlocality}} holds in the sense that the field operators commute
if their supports can be separated by a suitable spacelike surface layer (as shown in Figure~\ref{figurestrongCD}
on page~\pageref{figurestrongCD}).
It is one of the main objectives of this paper to make these statements mathematically precise
and to derive them from properties of minimizing measures of causal variational principles.
Moreover, based on the complex structure on the linearized fields
obtained in~\cite[Section~6.3]{fockbosonic}, distinguished
quasi-free states are constructed.

In order to clarify our concepts, we note that the construction of the algebra and
its Fock representation can be understood as a ``quantization'' of a classical bosonic field theory,
with the classical bosonic fields given by the linearized fields of the causal variational principle.
On the other hand, as worked out in~\cite{fockbosonic, fockfermionic},
the bosonic field algebra as well as a corresponding state also arise in the mathematical analysis of a
minimizing measure of an interacting causal variational principle.
In this sense, the interacting system does not need to be ``quantized,'' but
instead the causal variational principle already incorporates the
full quantum dynamics. We also point out that here we restrict attention to {\em{bosonic}} fields.
Fermionic fields could be treated similarly starting from the causal action principle for
causal fermion systems, working with the dynamical wave equation derived in~\cite{dirac}
and the fermionic algebra as well as the distinguished state constructed
in~\cite[Sections~4.1 and~4.5]{fockfermionic}.

The paper is organized as follows.
After the necessary preliminaries (Section~\ref{secprelim}), the classical theory of
linearized fields is developed, and the properties of the resulting classical dynamics are worked out
(Section~\ref{secclassical}). We proceed by constructing the algebra generated by these fields
and study its properties in comparison with the axioms of local quantum physics
(Section~\ref{secalgebra}). Finally, distinguished quasi-free states
are constructed (Section~\ref{secrepresent}).

\section{Preliminaries} \label{secprelim}
This section provides the necessary background on causal variational principles
and the linearized field equations.

\subsection{Causal Variational Principles in the Non-Compact Setting} \label{seccvp}
We consider causal variational principles in the non-compact setting as
introduced in~\cite[Section~2]{jet}. Thus we let~$\F$ be a (possibly non-compact)
smooth manifold of finite dimension~$m \geq 1$
and~$\rho$ a (positive) Borel measure on~$\F$.
Moreover, we are given a non-negative function~$\L : \F \times \F \rightarrow \R^+_0$
(the {\em{Lagrangian}}) with the following properties:
\bitem
\item[(i)] $\L$ is symmetric: $\L(x,y) = \L(y,x)$ for all~$x,y \in \F$.
\item[(ii)] $\L$ is lower semi-continuous, i.e.\ for all sequences~$x_n \rightarrow x$ and~$y_{n'} \rightarrow y$,
\[ \L(x,y) \leq \liminf_{n,n' \rightarrow \infty} \L(x_n, y_{n'})\:. \]
\eitem
The {\em{causal variational principle}} is to minimize the action
\beq \label{Sact} 
\Sact (\rho) := \int_\F d\rho(x) \int_\F d\rho(y)\: \L(x,y) 
\eeq
under variations of the measure~$\rho$, keeping the total volume~$\rho(\F)$ fixed
({\em{volume constraint}}).
Here the  notion {\em{causal}} in ``causal variational principles'' refers to the fact that
the Lagrangian induces on~$M$ a causal structure. Namely, two spacetime points~$x,y \in M$
are said to be {\em{timelike}} and {\em{spacelike}} separated if~$\L(x,y)>0$ and~$\L(x,y)=0$, respectively.
For more details on this notion of causality, its connection to the causal structure
in Minkowski space and to general relativity we refer to~\cite[Chapter~1]{cfs}, \cite{nrstg}
and~\cite[Sections~4.9 and~5.4]{cfs}.

If the total volume~$\rho(\F)$ is finite, one minimizes~\eqref{Sact}
within the class of all regular Borel measures with the same total volume.
If the total volume~$\rho(\F)$ is infinite, however, it is not obvious how to implement the volume constraint,
making it necessary to proceed as follows.
We need the following additional assumptions:
\bitem
\item[(iii)] The measure~$\rho$ is {\em{locally finite}}
(meaning that any~$x \in \F$ has an open neighborhood~$U$ with~$\rho(U)< \infty$).\label{Cond3}
\item[(iv)] The function~$\L(x,.)$ is $\rho$-integrable for all~$x \in \F$, giving
a lower semi-continuous and bounded function on~$\F$. \label{Cond4}
\eitem
Given a regular Borel measure~$\rho$ on~$\F$, we vary over all
regular Borel measures~$\tilde{\rho}$ with
\beq \label{totvol}
\big| \tilde{\rho} - \rho \big|(\F) < \infty \qquad \text{and} \qquad
\big( \tilde{\rho} - \rho \big) (\F) = 0
\eeq
(where~$|.|$ denotes the total variation of a measure).
For such variations, the difference of the actions~$\Sact(\tilde{\rho})-\Sact(\rho)$
is defined by
\begin{align*}
\big( &\Sact(\tilde{\rho}) - \Sact(\rho) \big) = \int_\F d(\tilde{\rho} - \rho)(x) \int_\F d\rho(y)\: \L(x,y) \\
&\quad + \int_\F d\rho(x) \int_\F d(\tilde{\rho} - \rho)(y)\: \L(x,y) 
+ \int_\F d(\tilde{\rho} - \rho)(x) \int_\F d(\tilde{\rho} - \rho)(y)\: \L(x,y) \:.
\end{align*}
If this difference is non-negative for all~$\tilde{\rho}$, then~$\rho$
is said to be a {\em{minimizer}}. 
The existence theory for such minimizers is developed in~\cite{noncompact}.
Moreover, in~\cite[Lemma~2.3]{jet} it is shown that a minimizer
satisfies the {\em{Euler-Lagrange (EL) equations}} which states
that, for a suitable value of the parameter~$\s>0$,
the lower semi-continuous function~$\ell : \F \rightarrow \R_0^+$ defined by
\beq \label{ldef}
\ell(x) := \int_\F \L(x,y)\: d\rho(y) - \s
\eeq
is minimal and vanishes on spacetime~$M:= \supp \rho$,
\beq \label{EL}
\ell|_M \equiv \inf_\F \ell = 0 \:.
\eeq
For more details we refer to~\cite[Section~2]{jet}.

\subsection{The Restricted Euler-Lagrange Equations} \label{secwEL}
The EL equations~\eqref{EL} are nonlocal in the sense that
they make a statement on~$\ell$ even for points~$x \in \F$ which
are far away from spacetime~$M$.
It turns out that, for the applications we have in mind, it is preferable to
evaluate the EL equations only locally in a neighborhood of~$M$.
This leads to the {\em{restricted EL equations}} introduced in~\cite[Section~4]{jet}.
We here give a slightly less general version of these equations which
is sufficient for our purposes. In order to explain how the restricted EL equations come about,
we begin with the simplified situation that the function~$\ell$ is smooth.
In this case, the minimality of~$\ell$ implies that the derivative of~$\ell$
vanishes on~$M$, i.e.\
\beq \label{ELrestricted}
\ell|_M \equiv 0 \qquad \text{and} \qquad D \ell|_M \equiv 0
\eeq
(where~$D \ell(p) : T_p \F \rightarrow \R$ is the derivative).
In order to combine these two equations in a compact form,
it is convenient to consider a pair~$\u := (a, u)$
consisting of a real-valued function~$a$ on~$M$ and a vector field~$u$
on~$T\F$ along~$M$, and to denote the combination of 
multiplication and directional derivative by
\begin{equation} \label{Djet}
\nabla_{\u} \ell(x) := a(x)\, \ell(x) + \big(D_u \ell \big)(x) \:.
\end{equation}
Then the equations~\eqref{ELrestricted} imply that~$\nabla_{\u} \ell(x)$
vanishes for all~$x \in M$.
The pair~$\u=(a,u)$ is referred to as a {\em{jet}}.

In the general lower-continuous setting, one must be careful because
the directional derivative~$D_u \ell$ in~\eqref{Djet} need not exist.
Our method for dealing with this problem is to restrict attention to vector fields
for which the directional derivative is well-defined.
Moreover, we must specify the regularity assumptions on~$a$ and~$u$.
To begin with, we always assume that~$a$ and~$u$ are {\em{smooth}} in the sense that they
have a smooth extension to the manifold~$\F$. Thus the jet~$\u$ should be
an element of the jet space
\beq \label{Jdef}
\J := \big\{ \u = (a,u) \text{ with } a \in C^\infty(M, \R) \text{ and } u \in \Gamma(M, T\F) \big\} \:,
\eeq
where~$C^\infty(M, \R)$ and~$\Gamma(M,T\F)$ denote the space of real-valued functions and vector fields
on~$M$, respectively, which admit smooth extensions to~$\F$.

Clearly, the fact that a jet~$\u$ is smooth does not imply that the functions~$\ell$
or~$\L$ are differentiable in the direction of~$\u$. This must be ensured by additional
conditions which are satisfied by suitable subspaces of~$\J$,
which we now introduce.
First, we let~$\Gdiff$ be formed by those vector fields for which the
directional derivative of the function~$\ell$ exists,
\[ \Gdiff(M,T\F) = \big\{ u \in \Gamma(M, T\F) \;\big|\; \text{$D_{u} \ell(x)$ exists for all~$x \in M$} \big\} \:. \]
This gives rise to the jet space
\[ 
\Jdiff := C^\infty(M, \R) \oplus \Gdiff(M,T\F) \;\subset\; \J \:. 
\]
For the jets in~$\Jdiff$, the combination of multiplication and directional derivative
in~\eqref{Djet} is well-defined. 
Next, we choose a linear subspace~$\Jtest \subset \Jdiff$ with the properties
that its scalar and vector components are both vector spaces,
\[ \Jtest = \Ctest(M, \R) \oplus \Gtest \;\subseteq\; \Jdiff \:, \]
and that the scalar component is nowhere trivial in the sense that
\beq \label{Cnontriv}
\text{for all~$x \in M$ there is~$a \in \Ctest(M, \R)$ with~$a(x) \neq 0$}\:.
\eeq
Then the {\em{restricted EL equations}} read (for details see~\cite[(eq.~(4.10)]{jet})
\[ 
\nabla_{\u} \ell|_M = 0 \qquad \text{for all~$\u \in \Jtest$}\:. \]
For brevity, a solution of the restricted EL equations is also referred to as a {\em{critical measure}}.
We remark that, in the literature, the restricted EL equations are sometimes also referred to as
the {\em{weak}} EL equations. Here we prefer the notion ``restricted'' in order to avoid confusion
with weak solutions of these equations (as constructed in~\cite{linhyp}; see also
Section~\ref{secweak} below).
The purpose of introducing~$\Jtest$ is that it gives the freedom to restrict attention to the part of
information in the EL equations which is relevant for the application in mind.
For example, if one is interested only in the macroscopic dynamics, one can choose~$\Jtest$
to be composed of jets pointing in directions where the 
microscopic fluctuations of~$\ell$ are disregarded.

We conclude this section by introducing a few other jet spaces 
and by specifying differentiability conditions which will be needed later on.
When taking higher derivatives on~$\F$, we are facing the general difficulty
that, from a differential geometric perspective, one needs to introduce a connection on~$\F$.
While this could be done, we here use the simpler method that higher derivatives on~$\F$
are defined as partial derivatives carried out in {\em{distinguished charts}}.
More precisely, around each point~$x \in \F$ we select a distinguished chart and
carry out derivatives as partial derivatives acting on each tensor component in this chart.
We remark that, in the setting of causal fermion systems, an atlas of distinguished charts is provided
by the so-called symmetric wave charts (for details see~\cite[Section~6.1]{gaugefix} or~\cite[Section~3]{banach}).

We now define the spaces~$\J^\ell$, where~$\ell \in \N \cup \{\infty\}$ can be
thought of as the order of differentiability if the derivatives act  simultaneously on
both arguments of the Lagrangian:
\begin{Def} \label{defJvary}
For any~$\ell \in \N_0 \cup \{\infty\}$, the jet space~$\J^\ell \subset \J$
is defined as the vector space of test jets with the following properties:
\bitem
\item[\rm{(i)}] For all~$y \in M$ and all~$x$ in an open neighborhood of~$M$,
directional derivatives
\beq \label{derex}
\big( \nabla_{1, \v_1} + \nabla_{2, \v_1} \big) \cdots \big( \nabla_{1, \v_p} + \nabla_{2, \v_p} \big) \L(x,y)
\eeq
(computed component-wise in charts around~$x$ and~$y$)
exist for all~$p \in \{1, \ldots, \ell\}$ and all~$\v_1, \ldots, \v_p \in \J^\ell$.
\item[\rm{(ii)}] The functions in~\eqref{derex} are $\rho$-integrable
in the variable~$y$, giving rise to locally bounded functions in~$x$. More precisely,
these functions are in the space
\[ L^\infty_\loc\Big( M, L^1\big(M, d\rho(y) \big); d\rho(x) \Big) \:. \]
\item[\rm{(iii)}] Integrating the expression~\eqref{derex} in~$y$ over~$M$
with respect to the measure~$\rho$,
the resulting function (defined for all~$x$ in an open neighborhood of~$M$)
is continuously differentiable in the direction of every jet~$\u \in \Jtest$.
\eitem
\end{Def} \noindent
Here and throughout this paper, we use the following conventions for partial derivatives and jet derivatives:
\bitem
\itemD Partial and jet derivatives with an index $i \in \{ 1,2 \}$, as for example in~\eqref{derex}, only act on the respective variable of the function~$\L$.
This implies, for example, that the derivatives commute,
\[ 
\nabla_{1,\v} \nabla_{1,\u} \L(x,y) = \nabla_{1,\u} \nabla_{1,\v} \L(x,y) \:. \]
\itemD The partial or jet derivatives which do not carry an index act as partial derivatives
on the corresponding argument of the Lagrangian. This implies, for example, that
\[ \nabla_\u \int_\F \nabla_{1,\v} \, \L(x,y) \: d\rho(y) =  \int_\F \nabla_{1,\u} \nabla_{1,\v}\, \L(x,y) \: d\rho(y) \:. \]
\eitem
We point out that, in contrast to the method and conventions used in~\cite{jet},
{\em{jets are never differentiated}}.

We denote the $\ell$-times continuously differentiable test jets by~$\Jtest \cap \J^\ell$.
Moreover, compactly supported jets are denoted by a subscript zero, like for example
\beq \label{J0def}
\Jtest_0 := \{ \u \in \Jtest \:|\: \text{$\u$ has compact support} \} \:.
\eeq
In order to make sure that surface layer integrals exist (see Section~\ref{secosi} below), one needs
differentiability conditions of a somewhat different type (for details see~\cite[Section~3.5]{osi}):
\begin{Def} \label{defslr}
The jet space~$\Jtest$ is {\bf{surface layer regular}}
if~$\Jtest \subset \J^2$ and
if for all~$\u, \v \in \Jtest$ and all~$p \in \{1, 2\}$ the following conditions hold:
\begin{itemize}[leftmargin=2.5em]
\item[\rm{(i)}] The directional derivatives
\beq \nabla_{1,\u} \,\big( \nabla_{1,\v} + \nabla_{2,\v} \big)^{p-1} \L(x,y) \label{Lderiv1}
\eeq
exist.
\item[\rm{(ii)}] The functions in~\eqref{Lderiv1} are $\rho$-integrable
in the variable~$y$, giving rise to locally bounded functions in~$x$. More precisely,
these functions are in the space
\[ L^\infty_\text{\rm{loc}}\Big( L^1\big(M, d\rho(y) \big), d\rho(x) \Big) \:. \]
\item[\rm{(iii)}] The $\u$-derivative in~\eqref{Lderiv1} may be interchanged with the $y$-integration, i.e.
\[ \int_M \nabla_{1,\u} \,\big( \nabla_{1,\v} + \nabla_{2,\v} \big)^{p-1} \L(x,y)\: d\rho(y)
= \nabla_\u \int_M \big( \nabla_{1,\v} + \nabla_{2,\v} \big)^{p-1} \L(x,y)\: d\rho(y) \:. \]
\eitem
\end{Def} \noindent
The precise regularity assumptions needed for our applications will be specified below
whenever we need them.

\subsection{The Linearized Field Equations}
The EL equations~\eqref{EL} (and similarly the restricted EL equations~\eqref{ELrestricted})
are nonlinear equations because they involve the measure~$\rho$
in a twofold way: first, the measure comes up as the integration measure in~\eqref{ldef},
and second, the function~$\ell$ is evaluated on the support of this measure.
Following the common procedure in mathematics and physics, one can simplify the
problem by considering linear perturbations about a given solution.
Demanding that these linear perturbations preserve the EL equations gives rise
to the {\em{linearized field equations}}.
More precisely, in our context we consider families of measures
which satisfy the restricted EL equations.
In order to obtain these families of solutions, we want to vary a given
measure~$\rho$ (typically a minimizer or a solution of the restricted EL equations)
without changing its general structure.
To this end, we multiply~$\rho$ by a weight function
and apply a diffeomorphism, i.e.
\beq \label{rhoFf}
\tilde{\rho} = F_* \big( f \,\rho \big) \:,
\eeq
where~$F \in C^\infty(M, \F)$ and~$f \in C^\infty(M, \R^+)$ are smooth mappings
(as defined before~\eqref{Jdef}). A variation of~$\rho$ is described by a
family~$(f_\tau, F_\tau)$ with~$\tau \in (-\delta, \delta)$ and~$\delta>0$. Infinitesimally, the variation
is again described by a jet
\beq \label{vinfdef}
\v = (b,v) := \frac{d}{d\tau} (f_\tau, F_\tau) \big|_{\tau=0}\:.
\eeq
The property of the family of measures~$\tilde{\rho}_\tau$ of the form~\eqref{rhoFf}
to be critical for a family~$(f_\tau, F_\tau)$ for all~$\tau$
means infinitesimally in~$\tau$ that the jet~$\v$ defined by~\eqref{vinfdef}
satisfies the {\em{linearized field equations}} (for the derivation see~\cite[Section~3.3]{perturb}
and~\cite[Section~4.2]{jet})
\beq \label{eqlinlip}
\la \u, \Delta \v \ra|_M = 0 \qquad \text{for all~$\u \in \Jtest$} \:,
\eeq
where
\beq \label{eqlinlip2}
\la \u, \Delta \v \ra(x)
:= \nabla_{\u} \bigg( \int_M \big( \nabla_{1, \v} + \nabla_{2, \v} \big) \L(x,y)\: d\rho(y) - \nabla_\v \:\s \bigg) \:.
\eeq
In order for the last expression to be well-defined, we always assume that~$\v \in \J^1$.
We denote the vector space of all solutions of the linearized field equations by~$\Jlin \subset \J^1$.

\subsection{The Inhomogeneous Linearized Field Equations} \label{secweaklinear}
For the analysis of the linearized field equations~\eqref{eqlinlip}, it is preferable to 
allow for an inhomogeneity~$\w$.
One method is to regard the inhomogeneity as a vector in the dual space of~$\Jtest$,
making it possible to insert on the right hand side of~\eqref{eqlinlip} the dual pairing~$\la \u, \w\ra|_M$
(for details see~\cite[Section~2.3]{linhyp}). For simplicity, we here avoid dual jets
and introduce instead a scalar product on the test jets, making it possible to identify
jets with dual jets. This procedure is of advantage also because the scalar product
on the jets will be needed later on for the construction of weak solutions.

Thus we choose a {\em{Riemannian metric}}~$g$ on the manifold~$\F$.
We denote the subspace spanned by the
test jets at the point~$x\in M$ by
\[ 
 \J_x := \big\{ \u(x) \:|\: \u \in \Jtest \big\} = \R \times \Gamma_x \;\subset\;  \R\times T_x\F\:. \]
The Riemannian metric induces a metric on~$\J_x$ by
\beq \label{vsprod}
\la \v, \tilde{\v} \ra_x := b(x)\, \tilde{b}(x) + g_x \big(v(x),\tilde{v}(x) \big) \:.
\eeq
We denote the corresponding norm by~$\|.\|_x$.
We point out that the choice of the Riemannian metric is {\em{not canonical}}.
The freedom in choosing the Riemannian metric can be used in order to satisfy the
hyperbolicity conditions needed for proving existence of solutions (as explained
after~\cite[Definition~3.3]{linhyp}). Since we will use these existence results
later on, we assume that the Riemannian metric in~\eqref{vsprod} has been chosen in agreement with the hyperbolicity conditions.

Having a scalar product at our disposal, we can formulate the inhomogeneous equations
by modifying~\eqref{eqlinlip} to
\beq \label{inhom}
\la \u, \Delta \v \ra(x) = \la \u, \w \ra_x \qquad \text{for all~$\u \in \Jtest$ and~$x \in M$}
\eeq
with an inhomogeneity~$\w \in \Jtest$. The linear functional~$\u(x) \mapsto \la \u, \Delta \v \ra(x)$
on the left hand side of this equation uniquely determines a vector~$(\Delta \v)(x) \in \J_x$ via the relation
\beq \label{frechet}
\la \u(x), (\Delta \v)(x)\ra_x = \la \u, \Delta \v \ra(x) \qquad \text{for all~$\u \in \Jtest$}
\eeq
(here we use the Fr{\'e}chet-Riesz theorem in finite dimensions).
Using this identification, we can also formulate the inhomogeneous equations in the shorter form
\beq \label{strong1}
\Delta \v = \w \qquad \text{on~$M$} \:.
\eeq
In order to ensure that the left side of this equation is well-defined,
we always assume (exactly as in~\eqref{eqlinlip}) that the jet~$\v$ lies in~$\J^1$.
We also point out that, using that the scalar components of the test jets is nowhere trivial~\eqref{Cnontriv},
the inhomogeneous equation~\eqref{inhom} implies that the scalar components in~\eqref{strong1} 
on both sides coincide pointwise. The same is true for the vector components, if one keeps in
mind that, by definition~\eqref{frechet}, the vector component of~$(\Delta \v)(x)$ always lies in~$\Gamma_x$.

\subsection{Linearized Fields in a Simple Example and in the Physical Context} \label{secexamples}
We now explain the concept of linearized fields and outline the connection to the physical applications.
We begin with a simple mathematical example first given in~\cite[Section~5.2]{static}
(see also~\cite[Section~20.2]{intro}). 
We let~$\F=\R^2$ and choose the Lagrangian as
\beq \label{gauss2}
\L(x,y; x',y') = \frac{1}{\sqrt{\pi}}\: e^{-(x-x')^2} \big( 1 + y^2 \big)\big( 1 +y'^2 \big) \:,
\eeq
where~$(x,y), (x',y') \in \F$. Using Fourier transform in the first variable
and a direct estimate, one sees that the measure
\beq \label{mingauss2}
d\rho = dx \times \delta_y
\eeq
(where~$\delta_y$ is the Dirac measure)
is the unique minimizer of the causal action principle under variations of
finite volume~\eqref{totvol} (for details see~\cite[Lemma~5.2]{static}
or~\cite[Lemma~20.2.1]{intro}). This result can be understood directly from the
fact that the function~$1+y^2$ is minimal at~$y=0$, which leads the minimizing measure~$\rho$
to be supported at~$y=0$. The Gaussian in the variable~$x$, on the other hand,
can be understood as a repelling potential, giving rise to a uniform distribution of the form
of the Lebesgue measure~$dx$.

The function~$\ell$ in~\eqref{ldef} is computed by
\[ \ell(x,y) = \int_\F \L(x,y; x',y') \: d\rho(x',y') - \s = 1+ y^2 -\s \:. \]
Therefore, the EL equations~\eqref{EL} are indeed satisfied for~$\s=1$.
{\em{Linearized fields}} describe linear perturbations of the measure of the form~\eqref{rhoFf}.
According to~\eqref{vinfdef}, they can be described by a jet~$\v$, consisting of a scalar function~$b$
and a vector field~$v$. In our example,
\beq \label{Jdiffex}
\Jdiff = \J = C^\infty(\R) \oplus C^\infty(\R, \R^2) \:,
\eeq
where~$C^\infty(\R, \R^2)$ should be regarded as the space of two-dimensional
vector fields along the $x$-axis. For clarity, we point out that the vector field~$v$
does not need to be tangential to~$M$. Instead, if the second component of~$v$
is non-zero, the vector field~$v$ is transversal to~$M$.
In this case, the first variation also describes an infinitesimal change of the support of the measure.
The linearized field equations~\eqref{eqlinlip} read
\beq \label{lininter}
\nabla_{\u} \bigg( \int_{-\infty}^\infty \big( \nabla_{1, \v} + \nabla_{2, \v} \big) 
e^{-(x-x')^2} \big( 1 + y^2 \big)\big( 1 +y'^2 \big)
\: dx' - \nabla_\v \,\sqrt{\pi} \bigg) \bigg|_{y=y'=0} = 0 \:.
\eeq
This is an integral equation which tells us about the freedom in perturbing the measure
while preserving the EL equations.

The prime example of a causal variational principle is the {\em{causal action principle}}
for causal fermion systems. This physical example also serves as the motivation
for physical notions like ``spacetime'' or ``fields.'' In order to avoid excessive repetitions
with previous work, here we shall not enter the definition of the causal action principle
and its connection to causal variational principles,
but refer instead to~\cite[Section~2]{dirac} or to the text books~\cite{cfs, intro}.
Instead, we only make a few general remarks with should convey the correct qualitative picture.
In simple examples of causal fermion systems describing a physical system in
Minkowski space or in a globally hyperbolic manifold, the set~$\F$ is a manifold
of very high dimension (possibly even infinite-dimensional), whereas spacetime~$M$ 
is a four-dimensional manifold. This means that the support of the measure~$\rho$
is highly singular, in the sense that it is supported on a subset of~$\F$ of very large co-dimension.
The situation is a bit similar to that in the above example~\eqref{mingauss2},
except that the number of transverse directions is not one, but very large.
The linearized fields describe infinitesimal changes of the measure in all these directions.
More specifically, in a causal fermion system the set~$\F$ consists of linear operators on
a Hilbert space~$(\H, \la .|. \ra_\H)$. In simple physical examples, one chooses~$\H$ as a subspace of the
Hilbert space of solutions of the Dirac equation in Minkowski space or in
a globally hyperbolic Lorentzian spacetime
(for details see for example~\cite[Section~5.4]{intro} or~\cite[Section~1]{nrstg}).
To every point~$x \in \scrM$ of our classical spacetime, we associate
a spacetime point operator~$F(x) \in \F$ constructed as 
the local correlation operator of the Dirac wave functions, i.e.\ it is defined by the relations
\[ \la \psi \,|\, F(x)\, \phi \ra_\H := - \Sl \psi(x) | \phi(x) \Sr_x \qquad \text{for all~$\psi, \phi \in \H$} \:, \]
where~$\Sl .|. \Sr_x$ denotes the inner product on the Dirac spinors at~$x$.
Consequently, the vector~$v(x)$ describing first variations of the operator~$F(x)$
can be described by first variations of the Dirac wave functions.
Thus the vector field~$v$ of a jet~$\v=(b,v)$ corresponds to a first variation of all the Dirac wave functions
of the system. Likewise, the linearized field equations tell us which first variations of the wave functions
preserve the EL equations of the causal action principle.
In order to describe first variations of Dirac wave functions, it is most convenient to insert
a bosonic potential~$\B$ into the Dirac equation,
\[ \big( \Dir + \B - m ) \psi = 0 \:. \]
First variations of the bosonic potential~$\B$ give rise to first variations of the Dirac wave functions,
which in turn give rise to the vector field~$v$ of a corresponding jet~$\v$.
Here the potential~$\B$ can be chosen as arbitrary operator acting on Dirac wave functions.
For example, it can be chosen to describe an electromagnetic or gravitational field.
In this formulation, it becomes clear that the linearized field equations are a very general concept
which allows to describe arbitrary physical fields of any spin.

\section{The Classical Dynamics of Linearized Fields} \label{secclassical}
In this section we recall a few properties of solutions of the linearized field equations.
We report on results first obtained in~\cite{linhyp} combined with methods developed similarly
for the dynamical wave equation in~\cite{dirac}. Our presentation is less general than in~\cite{linhyp}.
Instead, we aim at presenting a setting which is convenient but nevertheless sufficiently general for
the applications in mind. The main simplification compared to~\cite{linhyp} is that we assume the
existence of a global foliation by surface layers. But we show that our results are independent of the
choice of this foliation. In order to ``localize'' the solutions, we impose suitable shielding conditions.
These assumptions will be justified in the forthcoming paper~\cite{localize}.

\subsection{Global Foliations and Surface Layers} \label{secosi} $\;$
The linearized field operator~$\Delta$ in~\eqref{eqlinlip2} is a nonlocal integral operator.
In this paper, we shall always make the simplifying assumption that the range of this integral
operator is finite in the following sense.
\begin{Def}\label{compactrange}
The Lagrangian is said to have {\bf{compact range}} on~$M$ if for any compact~$K\subset M$ there are a
compact~$K'\subset M$ as well as open neighborhoods~$\Omega \supset K$ and~$\Omega' \supset K'$
of~$\F$ such that
$$
\L(x,y)=0 \quad \text{if~$x \in \Omega$ and~$y\not\in \Omega'$} \:.
$$
\end{Def} \noindent
For a variant of this definition and its usefulness we refer to~\cite{noncompact}.

Next, we shall assume that there is a {\em{global foliation}}. In analogy to
foliations by hypersurfaces of equal time in Lorentzian geometry, the idea is to cover spacetime by a family of surface layer integrals parametrized by a variable~$t$ which can again be thought of as the time
of a global observer.
\begin{Def}\label{defglobalfoliation}
A {\bf{global foliation}} is a family
\[ 
\eta \in \C^\infty(\R \times M, \R)\quad\mbox{with}\quad 0 \leq \eta \leq 1 \]
with the following properties:\\[-0.8em] 
\begin{itemize}[leftmargin=2.5em]
\item[{\rm{(i)}}] The function~$\theta(t,.) := \partial_t \eta(t,.)$ is non-negative. 
\item[{\rm{(ii)}}] The surface layers cover all of~$M$ in the sense that
\beq \label{2eq}
M = \bigcup_{t \in \R} \supp \,\theta(t,.) \:.
\eeq
\item[{\rm{(iii)}}] {\bf{Separation property}}: Let~$s_1,s_2\in\R$ and let~$K\subset M$ be open and relatively compact with the property that~$\eta(s_1, .)|_K \equiv 1$ and~$\eta(s_2,.)|_K\equiv 0$. Then
\[ \eta(s_1,x) \,\eta(s_2,x) = \eta(s_2,x) \qquad \text{for all~$x \in M$}\:. \] 
\end{itemize}
We also write~$\eta(t,x)$ as~$\eta_t(x)$ and~$\theta(t,x)$ as~$\theta_t(x)$.
\end{Def} \noindent
With this notion at our disposal, we can introduce {\bf{time strips}} simply by restricting the union in~\eqref{2eq}
to finite intervals,
$$
L_{s}^{t}:=\bigcup_{r\in [s,t]}\supp \,\theta(r,.) \:.
$$

Next, we introduce corresponding {\em{surface layer integrals}}.
In order to have the largest possible flexibility, we shall work with a subspace
\[ 
\Jvary \subset \Jtest \:, \]
which we can choose arbitrarily. 
Similar to~\eqref{J0def}, the space of jets in~$\Jvary$ with
compact support are denoted by~$\Jvary_0$.
For any~$t \in I$ we introduce the two bilinear forms
\begin{align}
(.,.)^t &\::\: \Jvary_0 \times \Jvary_0 \rightarrow \R \:, \notag \\
(\u, \v)^t &= \int_U d\rho(x)\: \eta_t(x) \int_U d\rho(y)\: \big(1-\eta_t(y)\big)
\: \Big( \nabla_{1,\u} \nabla_{1,\v} - \nabla_{2,\u} \nabla_{2,\v} \Big) \L(x,y) \label{jipdef} \\
\sigma^t(.,.) &\::\: \Jvary_0 \times \Jvary_0 \rightarrow \R \:, \notag \\
\sigma^t(\u,\v)&= \int_U d\rho(x)\: \eta_t(x) \int_U d\rho(y)\: \big(1-\eta_t(y)\big)
\: \Big( \nabla_{1,\u} \nabla_{2,\v} - \nabla_{1,\v} \nabla_{2,\u} \Big) \L(x,y) \label{sympdef} \:,
\end{align}
referred to as the {\em{surface layer inner product}} and the {\em{symplectic form}}, respectively.
In order to ensure that the integrals are well-defined, we assume throughout
this section that~$\Jtest$ is {\em{surface layer regular}} (see Definition~\ref{defslr}).
These surface layer integrals are ``softened versions'' of the surface layer integrals
introduced in~\cite{osi}.
In \cite{action} it was shown that for Dirac sea configurations in Minkowski space,
the inner product~$(.,.)^t$ is positive definite on the Dirac wave functions and on the Maxwell field tensor. With this in mind, it is sensible to assume that~$(\v,\v)^t$ is positive.
In a more quantitative form, this will be the content of the \textit{hyperbolicity conditions} below.

In order to introduce the hyperbolicity conditions we need to introduce the following bilinear operator~$\Delta_2$, which appears in the second variation of the causal action (for more details see \cite[Section 3.2]{osi}),
\[ \begin{split}
&\Delta_2: \Jvary_0 \times \Jvary_0 \rightarrow L^\infty_\loc(M, \R) \:,\\
&\Delta_2[\v_1,\v_2]:= \int_M \big( \nabla_{1, \v_1} + \nabla_{2, \v_1} \big)\big( \nabla_{1, \v_2} + \nabla_{2, \v_2} \big) \L(x,y)\: d\rho(y) - \nabla_{\v_1}\nabla_{\v_2} \:\s \:.
\end{split} \]
Then we have the following \textit{energy identity}, which can be proved as in~\cite[Lemma 3.2]{linhyp}.
\begin{Lemma} {\bf{(energy identity)}}
	For all~$\v\in \Jvary_0$,
	\begin{equation*}
	\frac{d}{dt}(\v,\v)^t=2\int_{M}\langle \v,\Delta \v\rangle\,d\rho_t-2\int_M \Delta_2[\v,\v]\,d\rho_t+\s \int_M b^2\,d\rho_t
	\end{equation*}
\end{Lemma}

We now introduce  a stronger and more quantitative version of positivity.
\begin{Def}\label{globfoliate}
	The global foliation in Definition~\ref{defglobalfoliation} is said to fulfill  the {\bf{hyperbolicity conditions}} if for every~$T>0$ there is a constant~$C>0$ such that
	\begin{equation*}
	(\v, \v)^t \geq \frac{1}{C} \int_M \Big( \|\v(x)\|_x^2\: + \big|\Delta_2[\v, \v]\big| \Big) \: d\rho_t(x)
	\label{hypcondM}
	\end{equation*}
	 for all~$t \in [-T,T]$ and all~$\v \in \Jvary_0$.
\end{Def} \noindent
The hyperbolicity conditions imply that the surface layer inner product at time~$t$ is
a scalar product. We denote the corresponding norm by
\[ 
\|\v\|^t := \sqrt{(\v,\v)^t} \:. \]

The next result can be proved as~\cite[Proposition 3.5]{linhyp}.
\begin{Prp} {\bf{(energy estimate)}} \label{prpees}
Let~$C$ be as in Definition~\ref{globfoliate}. Then, choosing
\[ \Gamma:= 2C e^{2C^2(1+\s^2/2)(t-s)}(t-s) \:, \]
the following estimate holds on the time strip~$L:=L_s^t$,
\beq \label{energyestimate}
\|\v\|_{L^2(L)}\le \Gamma \|\Delta \v\|_{L^2(L)}\qquad \text{for all~$\v\in \Jvary_0$ with~$\|\v\|^{s}=0$} \:.
\eeq
\end{Prp} \noindent
This result is the basis for the proof of existence of weak solutions on finite time strips, as will be
outlined in the next section.

\subsection{Weak Solutions of the Cauchy Problem}\label{secweak}
We now turn attention to the {\em{Cauchy problem}} for the linearized field equations
in a given time strip~$L:=L^t_s$.
It suffices to consider {\em{zero initial data}} at the initial time~$s$, because otherwise one can extend the initial
data to a jet~$\v_0$ in the time strip and consider the Cauchy problem for the difference~$\v - \v_0$
(for details on this standard procedure see for example~\cite[Section~3.5]{linhyp}).
It remains to be specified what it should mean for the solution to vanish at time~$s$.
The energy estimate~\eqref{energyestimate} suggests that one should demand that
the surface layer norm~$\|\v\|^s$ vanishes. For what follows, it is preferable
to work with the stronger condition that both surface layer integrals~\eqref{jipdef} and~\eqref{sympdef}
should vanish initially. Moreover, we demand that both~$\v$ and~$\Delta \v$ vanish
pointwise in the past of the surface layer at time~$t_0$. We thus introduce the jet space
\beq \label{Junder}
\begin{split}
\underline{\J}_s:=\big\{\v\in\Jvary_0\:|\:
\eta_s \v &= 0 = \eta_s \Delta\v \qquad \text{and}\\
(\u, \v )^s &=0 = \sigma^s(\u, \v) \qquad \text{for all~$\u\in \Jvary_0$} \big\} \:.
\end{split}
\eeq
Then a jet~$\v \in \Jvary_0$ is referred to as a {\em{strong solution}} of the Cauchy problem if
\beq \label{strong}
\text{$\Delta\v=\w$ in~$L$} \qquad \text{and} \qquad \v \in \underline{\J}_s \:.
\eeq
The energy estimate in Proposition~\ref{prpees} ensures that the strong solution is unique:
\begin{Prp}
	The Cauchy problem~\eqref{strong} has at most one solution.
\end{Prp}

For the existence theory, we need to formulate the Cauchy problem in the weak sense.
Similar as in the theory of partial differential equations, the weak formulation is obtained
by multiplying with a test function and integrating by parts.
In our context, we work with the following~$L^2$-scalar product in the time strip~$L:=L_s^t$,
\begin{equation*}
\la \u,\v\ra_{L^2(L)}:=\int_M \la \u,\v\ra_x\,\eta_{[s,t]}(x)\,d\rho(x),\quad L^2(L):=\big\{
\v\in L^2_\loc(M)\:\big|\: \|\v\|_{L^2(L)}<\infty \big\} \:.
\end{equation*}
with~$\eta_{[t_0,t_1]} := \eta_{t_1} - \eta_{t_0}$. The analog of ``integration by parts'' is provided by the
following lemma, which can be proved exactly as~\cite[Lemma 3.13]{linhyp}.

\begin{Lemma} {\bf{(Green's formula)}} \label{green}
For any time strip~$L:=L_s^t$ and any~$\u,\v\in\Jvary_0$,
\beq \label{eqgreen}
\la \u,\Delta \v\ra_{L^2(L)}=\la \Delta\u, \v\ra_{L^2(L)}-\sigma^{t}(\u,\v)+\sigma^{s}(\u,\v) \:.
\eeq
\end{Lemma} \noindent
Assume that we are given a strong solution~$\v\in \underline{\J}_s$. Then the Green's formula
shows that for any~$\u\in \Jvary_0$,
\[ \langle \u,\w\rangle_{L^2(L)}=\langle \u,\Delta \v\rangle_{L^2(L)}=\langle \Delta\u,\v\rangle_{L^2(L)}
-\sigma^{t}(\u,\v)+\sigma^{s}(\u,\v) \:. \]
Using the last equation in~\eqref{Junder}, the symplectic form at time~$s$ vanishes.
In order for the boundary term at time~$t$ to vanish, we choose the test jet~$\u$ in the space~$\overline{\J}^t$
defined in analogy to~\eqref{Junder} by
\beq \label{Jupper}
\begin{split}
\overline{\J}^t:=\big\{\v\in\Jvary_0\:|\:
(1-\eta_t) \v &=0= (1-\eta_t) \Delta\v \qquad \text{and}\\
( \u, \v )^t &=0 = \sigma^t(\u, \v) \quad \text{for all~$\u\in \Jvary_0$} \big\} \:.
\end{split}
\eeq
We thus obtain the equation
\begin{equation}\label{weak}
\langle \Delta\u, \v\rangle_{L^2(L)} = \langle \u,\w\rangle_{L^2(L)} \qquad
\mbox{for all }\u\in \overline{\J}^{t}.
\end{equation}
We take this equation as the {\em{definition}} of the weak solution of the Cauchy problem
(for a more detailed explanation see~\cite[Section~3.6]{linhyp}).
Since only~$L^2$-products are involved, one can allow for inhomogeneities and solutions which are merely square integrable.
\begin{Def}
Let~$\w\in L^2(L)$. A jet~$\v\in L^2(L)$ satisfying \eqref{weak} is said to be a {\bf{weak solution}} of \eqref{strong}.
\end{Def}
We now come to the existence proof of weak solutions. In preparation,
on~$\overline{\J}^t$ we introduce the scalar product
\[ \lla \u,\v\rra:=\langle \Delta\u,\Delta\v\rangle_{L^2(L)} \:. \]
Note that~$\lla .,. \rra$ is bilinear and positive definite. The non-degeneracy is a direct consequence of the energy estimate~\eqref{energyestimate} (with the time direction reversed), which yields
\beq \label{ees}
\|\v\|_{L^2(L)}\le \Gamma \|\Delta\v\|_{L^2(L)}=\tb \v\tb\quad\mbox{for all~$\v\in \overline{\J}^{t}$}.
\eeq
Forming the completion, we obtain a real Hilbert space, which we
denoted by~$(\overline{\mathcal{H}}^t, \lla .,. \rra)$ or, in order to clarify the dependence on the time strip,
by~$(\overline{\mathcal{H}}(L), \lla .,. \rra_L)$, i.e.\
\begin{equation}\label{energyspace}
\overline{\mathcal{H}}^t = \overline{\mathcal{H}}(L) :=\text{completion of } \big(\overline{\J}^t, \lla .,.\rra \big)\qquad \text{with} \qquad  \lla \u,\v\rra:=\langle \Delta\u,\Delta\v\rangle_{L^2(L)}.
\end{equation}

Using again the energy estimate~\eqref{ees}, this Hilbert space can be
characterized as follows: It is formed by all vectors~$V \in L^2(L)$
for which there is a sequence~$(V_n)_{n \in \N}$ in~$\overline{\J}^t$ with~$V_n \rightarrow V$ in~$L^2(L)$
and the sequence~$\Delta V_n$ is Cauchy in~$L^2(L)$. Denoting the limit of this Cauchy sequence by~$\Delta V$,
we conclude that the operator~$\Delta$ extends to a unique bounded linear operator
from~$\overline{\mathcal{H}}^t$ to~$L^2(L)$.

In view of this result, the linear functional~$\la \w , .\, \ra_{L^2(L)}$ extends to a
bounded linear functional on the Hilbert space~$\overline{\mathcal{H}}^t$. Representing this functional
with the help of the Fr\'echet-Riesz theorem by a vector~$V$ of this Hilbert space
gives the following result (for more details see~\cite[proof of Theorem~3.15]{linhyp}):
\begin{Thm} \label{thmexistence}
Given~$\w\in L^2(L)$, there exists a unique solution of~\eqref{weak} of the form~$\v = \Delta V$ with~$V\in \overline{\mathcal{H}}^t$.
\end{Thm} \noindent

\begin{Def} We refer to the distinguished solution in Theorem~\ref{thmexistence}
as the {\bf{retarded solution}}, denoted by
\beq \label{vret}
\v^\wedge(L,\w) := \Delta V \:.
\eeq
\end{Def}

Weak solutions of the Cauchy problem are in general not unique. This can be understood immediately from
the fact that in~\eqref{weak} we test only with a specific subspace of jets. More precisely,
let~$\v,\v'$ be two weak solutions. For clarity of our explanation, we first consider the case
that their difference~$\bar{\v}:= \v-\v'$ lies in in~$\underline{\J}_s$.
Applying Lemma~\ref{green} we obtain the homogeneous weak equation
\begin{equation}\label{vanishingL}
\langle \Delta \bar{\v},\u\rangle_{L^2(L)}=0\quad\mbox{for all }\u\in \overline{\J}^{t}.
\end{equation}
If the set~$\overline{\J}^{t}$ were dense in~$L^2(L)$, we could apply the energy
estimates~\eqref{energyestimate} to conclude
$$
\Delta\bar{\v}=0\ \quad\Longrightarrow\quad  \|\bar{\v}\|_{L^2(L)}\le C\|\Delta\bar{\v}\|_{L^2(L)}=0\
\quad\Longrightarrow \quad \v=\v'.
$$ 
However, denseness of~$\overline{\J}^{t}$ is not a sensible assumption. This reflects our general concept that, by
choosing specific subspaces of~$\J$, we restrict attention to the part of information of the EL equations which is relevant to the applications in mind. More generally, dropping the simplifying
assumption~$\v-\v' \in \underline{\J}_s$, a weak solution of~\eqref{weak}
is determined only up to vectors in the orthogonal complement of~$\Delta(\overline{\J}^{t})$,
\[ \v-\v'\in L^2(L)\cap \Delta(\overline{\J}^{t})^\perp \:. \]
The resulting freedom to modify a weak solution is irrelevant to us because it only affects the information
that we disregard.

\subsection{Construction of Global Weak Retarded Solutions}
In this section we want to apply the existence theory in time strips as outlined in the previous section 
in order to construct \textit{global retarded weak solutions} in~$L^2_\loc(M)$ of the equation
\begin{equation}\label{strong2}
\Delta \v=\w\quad\mbox{for~$\w\in L^2_0(M)$}.
\end{equation}
In order to make sense of this equation, let us test \eqref{strong2} with compactly supported jets and then apply Lemma~\ref{green}. This yields the weak equation
\begin{equation}\label{weakequation}
	\langle \Delta \u,\v\rangle = \langle \u,\w\rangle\quad\mbox{for all~$\u\in \Jvary_0$}.
\end{equation}
This equation is well-defined and leads us to the following definition.
\begin{Def}
	Let~$\w\in L_0^2(M)$. A jet~$\v\in L^2_{\loc}(M)$ satisfying \eqref{weakequation} is said to be a {\bf{global weak solution}} of \eqref{strong2}.
\end{Def}

Our next goal is to construct global \textit{retarded} weak solutions. Roughly speaking, we aim at constructing a global weak solution which vanishes in the past of the inhomogeneity~$\w$.  
Our strategy is to consider the weak retarded solution constructed in Theorem~\ref{thmexistence}
in time strips~$L_{s}^{t}$ and to take the limits~$t \to \infty$ and~$s \to-\infty$.
In order to ensure that the limits exist,
we need to make further assumptions on the foliations and on the operator~$\Delta$. 
\begin{Def} \label{defproperlycontained}
A set~$U\subset M$ is said to be {\bf{properly contained}} in a time strip~$L_s^t$, denoted
by~$U\stackrel{\circ}{\subset} L_{s}^{t}$, if
\[ (1-\eta_{s})|_U\equiv \eta_t|_U\equiv 1\:. \]
\end{Def}  

We now introduce the so-called \textit{shielding condition}.
It can be understood as a generalization  of the denseness of~$\overline{\J}^t$ in~$\overline{\mathcal{H}}^t$ to situations when various time strips are involved.

\begin{Def}\label{Defshielding}
Spacetime~$M$ is said to be {\bf{shielded in the future}} with respect to a given foliation
$(\eta_t)_{t \in \R}$ if for every~$L_s^t$ there is~$L_{s_1}^{t_1}\stackrel{\circ}{\supset} L_s^t$
such that, for every~$L_{s_2}^{t_2}\stackrel{\circ}{\supset} L_{s_1}^{t_1}$
the following implication holds for all~$V_2\in\overline{\mathcal{H}}(L_{s_2}^{t_2})$ and~$V_1\in\overline{\mathcal{H}}(L_{s_1}^{t_1})$:
\begin{equation*}
	\begin{split}
		&\langle \Delta V_2+(1-\eta_{s_1})\Delta V_1,\Delta\u\rangle_{L^2(L_{s_2}^{t_1})}=0\quad\mbox{for all }\ \u\in\overline{\J}^{t_1}\\
		&\qquad\qquad\Longrightarrow \quad \Delta V_2+(1-\eta_{s_1})\Delta V_1\equiv 0\quad\mbox{on }\ L_{s_2}^t.
	\end{split}
\end{equation*}
\end{Def} \noindent
This condition has the consequence that the weak solutions constructed in Theorem~\ref{thmexistence} can be extended consistently to global solutions. Moreover, the resulting global solutions vanish in the past of the inhomogeneity.
This is made precise in the following theorem.

\begin{Thm}\label{global}
Assume that spacetime~$M$ is shielded in the future with respect to a foliation~$(\eta_t)_{t \in \R}$. Then the following
statements hold:
\begin{itemize}[leftmargin=2.5em]
\item[\rm{(i)}] 	For any~$\w\in L_0^2(M)$ the following limit 
of retarded solutions~\eqref{vret} exists in the topology of~$L_\loc^2(M)$,
\begin{equation}\label{limitloc}
	-\lim_{s\to -\infty}\lim_{t\to+\infty}\v^\wedge(\w,L_s^t)=: S_\eta^\wedge\w
\end{equation}
\item[\rm{(ii)}]  The jet~$S_\eta^\wedge\w$ is a global weak solution of \eqref{strong2}, i.e.
$$
\langle S_\eta^\wedge\w,\Delta\u\rangle_{L^2(M)}=-\langle \w,\u\rangle_{L^2(M)}\quad\mbox{for all~$\u\in\Jvary_0$}.
$$\\[-2.1em]
\item[\rm{(iii)}] The jet~$S_\eta^\wedge\w$ vanishes in the distant past in the sense that
there exists~$t \in \R$ for which the following implication holds:
\[ \eta_t(x)=1 \qquad \Longrightarrow \qquad \w(x)=0 \quad \text{and} \quad (S_\eta^\wedge \w)(x)=0 \:. \]
\end{itemize}
\end{Thm}
\begin{proof}
We choose~$L_{s}^{t}$ so large that~$\supp\w\stackrel{\circ}{\subset} M$, choose~$L_{s_1}^{t_1}$ to be as in Definition~\ref{Defshielding} and let~$L_{s_2}^{t_2}\stackrel{\circ}{\supset} L_{s_1}^{t_1}$ be arbitrary. 
Finally, let~$\u\in \overline{\J}^{t_1}\subset \overline{\J}^{t_2}$. Then, using the notation of Theorem~\ref{thmexistence} and the fact that (cf. \eqref{Jupper})
\begin{align*}
(1-\eta_{s_2})(1-\eta_{s_1}) &=(1-\eta_{s_1}) \:, &  (1-\eta_{s_1})\w&=\w= (1-\eta_{s_2})\w,\\ 
\eta_{t_1}\u &=\u=\eta_{t_2}\u \:,& \eta_{t_1}\Delta\u&=\Delta\u=\eta_{t_2}\Delta\u \:,
\end{align*}
one readily finds that
\begin{align*}
\langle \v(L_{s_2}^{t_2},\w),\Delta\u\rangle_{L^2(L_{s_2}^{t_1})}&=\langle \v(L_{s_2}^{t_2},\w),\Delta\u\rangle_{L^2(L_{s_2}^{t_2})}=\langle \w,\u\rangle_{L^2(L_{s_2}^{t_2})}=\langle \w,\u\rangle_{L^2(L_{s_1}^{t_1})} \\
&=\langle\v(L_{s_1}^{t_1},\w),\Delta\u\rangle_{L^2(L_{s_1}^{t_1})}= \langle (1-\eta_{s_1})\v(L_{s_1}^{t_1},\w),\Delta\u\rangle_{L^2(L_{s_2}^{t_1})} \:.
\end{align*}
It follows that
\[ \langle \v(L_{s_2}^{t_2},\w)-(1-\eta_{s_1})\v(L_{s_1}^{t_1},\w),\Delta\u\rangle_{L^2(L_{s_2}^{t_1})}=0\quad\mbox{for all }\u\in \overline{\J}^{t_1} \:. \]
From Theorem \eqref{thmexistence} we know that every~$\v(L,\w)$ 
can be written in in the form~$\v(L,\w)=\Delta V$ for suitable~$V\in\mathcal{H}(L)$.
Therefore, the shielding condition in Definition~\ref{Defshielding} yields
\begin{equation}\label{finalidentity}
\v(L_{s_2}^{t_2},\w)-(1-\eta_{s_1})\v(L_{s_1}^{t_1},\w)=0\quad\mbox{on~$L_{s_2}^{t}$}.
\end{equation}
This identity has two consequences. First, it follows that
$
\v(L_{s_2}^{t_2},\w)\equiv \v(L_{s_1}^{t_1},\w)
$
on~$L_s^t$. Since~$s$ and~$t$ are arbitrary,
we conclude that for any compact~$K\subset M$ there are sufficiently large~$\bar{a},\bar{b}$ such that 
$$
\v(L_{a}^{b},\w)|_K= \v(L_{\bar{a}}^{\bar{b}},\w)|_K \mbox{ for all~$a \leq \bar{a}$ and~$b\ge\bar{b}$}.
$$  
Hence the limit \eqref{limitloc} is well-defined. 
The second consequence of \eqref{finalidentity} is that 
$$
\v(L_{s_2}^{t_2},\w)\equiv 0\quad\mbox{on}\quad \{\eta_{s_1}\equiv 1 \}\cap \{\eta_{s_2}\equiv 0\}.
$$
From the arbitrariness of~$s_2$ and~$t_2$ we infer  that also the limit function~$S_\eta^\wedge\w$ vanishes
in the region~$\{\eta_{s_1}\equiv 1 \}$.

To summarize, the limit function~$S_\eta^\wedge\w$ defines a global weak solution of \eqref{strong2}. This follows directly from the fact that every~$\v^\wedge(\w,L_{s}^{t})$ is a weak solution within the corresponding time strip~$L_{s}^{t}$, and the fact that~$s$ and~$t$ can be chosen arbitrarily. 
\end{proof}

We conclude this section by a refinement of condition~(iii) in Theorem~\ref{global}.
We begin with the following elementary observation.
\begin{Lemma}\label{lemmacompact}
For every open and relatively compact subset~$W \subset M$ there exists a compact set~$\mathfrak{K}_\eta(W)$ with the following property:
\[ \eta_{s} \big|_{\mathfrak{K}_\eta(W)}\equiv 0 \qquad \Longrightarrow \qquad
\text{$\eta_{s} \,S_\eta^\wedge\,\w=0$ for all~$\w\in L_0^2(W)$}\:. \]
\end{Lemma}
\begin{proof}
Let~$L_{s}^{t}\stackrel{\circ}{\supset} W$ be arbitrary and let~$L_{s_1}^{t_1}\stackrel{\circ}{\supset}L_{s}^{t}$ be as in Definition~\ref{Defshielding}. We now choose a sufficiently large compact set~$\mathfrak{K}_\eta(W)\supset W$ so that
\[ K:=\mathrm{int}(\mathfrak{K}_\eta(W)\cap \{ \eta_{s_1}\equiv 1\})\neq \varnothing \:. \]
Now, let~$s_2\in\R$ be such that~$\eta_{s_2}|_{\mathfrak{K}_\eta(W)}\equiv 0$. By construction, we have $\eta_{s_1}|_K\equiv 1$ and $\eta_{s_2}|_{K}\equiv 0$.  Then,  applying Definition~\ref{defglobalfoliation}~(iii) to the relatively compact open set $K$, we infer that~$\eta_{s_1}\eta_{s_2}=\eta_{s_2}$, and hence~$\eta_{s_2}(x)>0\Rightarrow\eta_{s_1}(x)=1$. 
Now let~$\w\in L_0^2(M)$ with~$\supp\w\subset W$. Then, from the proof of Theorem~\ref{global} we see that the global solution~$S_\eta^\wedge\w$ vanishes in the region~$\{\eta_{s_1}\equiv 1\}$. From the discussion above we conclude that~$\eta_{s_2}S^\wedge_\eta\w\equiv 0$.
\end{proof}

We point out that the above construction of the compact set~$\mathfrak{K}_\eta(W)$ depends on the chosen foliation. However, it only involves the ``width'' of the surface layers
given by the support of the functions~$\theta_t$. In physical application, this ``width''
is of the order of the Compton length. The set~$\mathfrak{K}_\eta(W)$ can be thought of as a
``neighborhood'' of~$W$ containing a boundary strip on the Compton scale.
With this in mind, by restricting attention to a subfamily of foliations if necessary,
it is sensible to make the following assumption:
\begin{Assumption}\label{remarkcompact}
The compact set~$\mathfrak{K}_\eta(W)$ can be chosen uniformly in the foliation~$\eta$. For this reason, the lower index~$\eta$ will henceforth be omitted.
\end{Assumption}

\subsection{Independence of Foliations and Green's Operators}\label{sectionfoliation}  The global weak solutions constructed in Definition~\ref{global} depend on the choice of foliation~$(\eta_t)_{t \in \R}$.
The goal of this section is to work out additional assumptions under which the
solutions become independent of this choice.

Recall from the definition of energy space~\eqref{energyspace}
that for any vector~$V\in \overline{\mathcal{H}}^t$ there is a sequence~$\u_n\in\overline{\J}^t$ and a jet~$\Delta V\in L^2(L)$ such that~$\u_n\to V$ and~$\Delta\u_n\to \Delta V$ in~$L^2(L)$. Moreover, 
the following implications hold,
\begin{equation}\label{vanishingcond}
\langle \Delta\u,\Delta V\rangle_{L^2(L)}=0\quad\ \text{for all~$\u\in \overline{\J}^t$}
\quad \Longrightarrow \quad   V = 0 \quad \Longrightarrow \quad \Delta V = 0 \:.
\end{equation}
In order to clarify the significance of these implications,
we remark that the condition on the very left combines two properties of the jet~$\Delta V$. On the one hand, restricting to test functions~$\u\in \overline{\J}^t$ which are supported in the interior of~$L$ and formally applying the Green's
formula yields a condition on the interior values of~$\Delta \Delta V$. On the other hand, restricting to test functions which intersect the surface layer integral at initial time yields a condition on the initial values of~$\Delta V$.
These conditions taken together imply that~$\Delta V=0$. This is the content of the implications in~\eqref{vanishingcond}.

Now let~$U,U_1\in L^2_\loc(M)$ and let~$\u_n\in \Jvary_0$ be such that~$\u_n\to U$ and~$\Delta\u_n\to U_1$
in~$L^2_\loc(M)$. 
The function~$U$ has the same structure as the function~$V$ above, except that it is defined in all of
spacetime~$M$. One may then expect that a condition of the form 
$
\langle \Delta u, U_1\rangle_{L^2(M)}=0
$
for all~$\u\in \Jvary_0$,
together with the property that~$U_1$ vanishes identically in the past of some~$\eta_t$ would imply~$U_1=0$.
However, this implication is not straightforward, although it could be arranged to hold by introducing additional structures on spacetime (for example, with the methodology of cutoff operators introduced
in~\cite[Section 6.6]{dirac} one obtains new conservation laws which can be used to achieve this;
for more details see~\cite{localize}).
For the purposes of this paper and for simplicity of presentation, it seems sensible to
simply condense this implication into a new condition.
\begin{Def}\label{completcondi}
Spacetime~$M$ is said to fulfill the {\bf{completeness condition}} (in the future) if the following property holds: \\[0.3em]
Let~$U_1\in L^2_\loc(M)$ and let~$\u_n\in \Jvary_0$ be a Cauchy sequence in~$L^2_\loc(M)$ with the property that~$\Delta\u_n\to U_1$ in~$L^2_\loc(M)$. Then, 
for every two foliations~$\eta^1,\eta^2$ and~$t_1,t_2\in\R$,
\begin{equation*}
\left. 
\begin{split}
	\langle \Delta\u,U_1\rangle_{L^2(M)}=0\  \mbox{ for all }\ \u\in \Jvary_0\\
	U_1\equiv 0\ \mbox{ on }\ \{ \eta^1_{t_1}\equiv 1\}\cap \{ \eta^2_{t_2}\equiv 1\}
\end{split}	
\ \right\}
\ \Longrightarrow\ \  U_1=0.
\end{equation*} 
\end{Def}

We will now show how this completeness assumption can be used to construct
causal Green's operators which are independent of the choice of the foliation.
We begin with the following preparatory observation.
Bearing in mind the construction of global solutions of Theorem~\ref{global} and 
applying a diagonal sequence argument to a sequence of time strips exhausting spacetime,
one readily verifies that
\[ 
\begin{split}
&\text{For any~$\w\in L^2_{0}(M)$ there is a sequence~$\u_n\in\Jvary_0$ which is } \\[-0.2em]
&\text{Cauchy in~$L^2_\loc(M)$ and satisfies~$\Delta\u_n\to S^\wedge_\eta\w$ in~$L^2_\loc(M)$}\:.
 \end{split} \]
We now consider the global weak solutions~$S_1^\wedge\w$ and~$S_2^\wedge\w$  associated with two different foliations~$\eta^1,\eta^2$, respectively. Let~$U_1:=S_1^\wedge\w-S_2^\wedge\w\in L^2_\loc(M)$. By definition of global weak solutions, we know that
\[ 
\la U_1,\Delta\u\rangle_{L^2(M)}=\langle S_1^\wedge\w,\Delta\u\rangle_{L^2(M)} - \langle S_2^\wedge\w,\Delta\u\rangle_{L^2(M)}= 0\quad\mbox{for all~$\u\in \Jvary_0$} \:. \]
Moreover, from Theorem~\ref{global}~(iii) there exist~$t_1$ and~$t_2$ such that 
$
U_1=0
$
on the intersection~$\{ \eta^1_{t_1}\equiv 1\}\cap \{\eta_{t_2}^2\equiv 1 \}$. The completeness condition in Definition~\ref{completcondi} implies that~$U_1\equiv 0$, so that~$S^\wedge_1\w = S^\wedge_2\w$ as desired. In summary, we have the following result.
\begin{Prp}\label{equalityS}
The global weak solution of Theorem~\ref{global} does not depend on the foliation. The index~$\eta$ will be dropped accordingly.
\end{Prp}

The global solution~$S_\eta^\wedge\w$ is said to be \textit{retarded} in the sense that it vanishes in the past of the inhomogeneity~$\w$. In a similar way, reversing the direction of time (and adjusting Definitions~\ref{Defshielding} and~\ref{completcondi} accordingly), one can construct corresponding \textit{advanced} solutions~$S_\eta^\vee\w$ which vanish instead in the future of~$\w$. 

By restricting to arbitrary time strips and using Theorem~\ref{global} one sees that the mappings~$S_\eta^\vee$ and~$S_\eta^\wedge$ depend \textit{linearly} on the inhomogeneity.
\begin{Def} \label{defretarded}   The {\bf retarded} and {\bf advanced Green's operators} are defined, respectively, as the linear mappings
\[ S^\wedge \::\: L^2_{0}(M) \rightarrow L_\loc^2(M)\,,\qquad 
S^\vee \::\: L^2_{0}(M) \rightarrow L_\loc^2(M).\:. \]
Their difference
\beq \label{Gdef}
G := S^\wedge -S^\vee \::\: L^2_{0}(M)\rightarrow L_\loc^2(M)
\eeq
is referred to as the {\bf causal fundamental solution}.
\end{Def} \noindent
By construction, the causal fundamental solution~$G$ maps compactly supported, square-integrable jets to weak solutions of the \textit{homogeneous equation}, i.e.
$$
\langle G\w,\Delta\u\rangle_{L^2(M)}=0\quad\mbox{for all~$\u\in\Jvary_0$}.
$$

The goal of the next section is to give the following statement a precise mathematical meaning: \textit{The action of the retarded and advanced Green's operators is causal}.

\subsection{Causal Structure of the Linearized Fields and of Spacetime}
So far we have worked under the assumption that a global foliation fulfilling the hyperbolicity conditions of Definition~\ref{globfoliate} exists. Let~$\Theta$ denote the class of \textit{all} such global foliations. This allows us to introduce the following notion of \textit{causality}.
\begin{Def}
Let~$x,y\in M$ be two spacetime points. We say that:
\begin{itemize}[leftmargin=2.5em]
\item[\rm{(i)}] $x$ {\bf{chronologically precedes}}~$y$ and denote it by~$x\ll y$, if 
\[ \forall\, (\eta_t)_{t \in \R} \in\Theta\ \ \forall\, t\in\R:\quad\eta_t(x)<1\ \Longrightarrow\ \eta_t(y)=0 \:. \]
\item[\rm{(ii)}] $x$ {\bf{causally precedes}}~$y$, and denote it by~$x\prec y$, if
\[ \forall\, (\eta_t)_{t \in \R} \in\Theta\ \ \forall\, t\in\R:\quad\eta_t(x)=0\ \Longrightarrow\ \eta_t(y)<1 \:. \]
\end{itemize}
\end{Def}
One immediately verifies that~$\ll$ induces a \textit{transitive} relation on~$M$. This is, however, in general not true for the causal relation~$\prec$. On the other hand, while the relation~$\prec$ is \textit{reflexive}, the chronological relation~$\ll$ is not, i.e.
\begin{equation}\label{reflexive}
x \not\ll x\qquad \text{but} \qquad x\prec x.
\end{equation}
For each relation, we can now introduce a corresponding notion of future and past cone. The two different future cones are depicted in Figure~\ref{fig1}, which also illustrates~\eqref{reflexive}. 
\begin{figure}
\centering
\def\svgwidth{\columnwidth}
\scalebox{0.8}{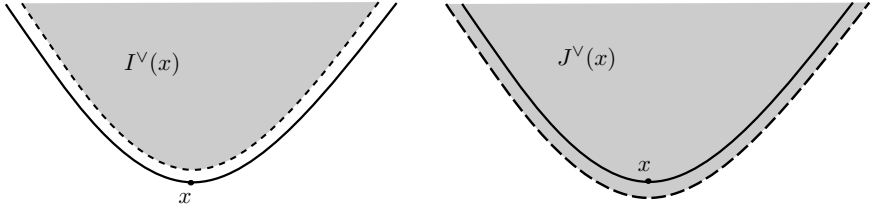}
\caption{Chronological and causal future cones.}
\label{fig1}
\end{figure}%

\begin{Def}
The {\bf{future}} and {\bf{past chronological cones}} are defined by
$$
I^\vee(x):=\left\{y\in M\:|\: x\ll y\right\} \quad\mbox{and}\quad
\quad I^\wedge(x):= \left\{y\in M\:|\: y\ll x \right\}.
$$
The {\bf{future}} and {\bf{past causal cones}} are defined by
$$
J^\vee(x):=\left\{y\in M\:|\: x\prec y\right\} \quad\mbox{and}\quad
\quad J^\wedge(x):= \left\{y\in M\:|\: y\prec x \right\}.
$$
The unions
$$
I(x):=I^\vee(x)\cup I^\wedge(x) \quad\mbox{and}\quad J(x):= J^\vee(x)\cup J^\wedge(x)
$$
are referred to as the {\bf{chronological}} and {\bf{causal cone}} of~$x$, respectively.
\end{Def} \noindent
It follows immediately from the definitions that
$$
x\in I^\vee(y)\ \mbox{ iff }\ y\in I^\wedge(x)\quad\mbox{and}\quad x\in J^\vee(y)\ \mbox{ iff }\ y\in J^\wedge(x) \:.
$$
Statement \eqref{reflexive} can be rephrased in terms of future and past cones as follows,
$$
x\in J^\vee(x) \qquad\mbox{but}\qquad x \not\in I^\vee(x)
$$
(and similarly for the past cones).
In a similar way, one can define future and past cones generated by a compact subset~$K\Subset M$ by
$$
J^\vee(K):=\bigcup_{x\in K}J^\vee(x)\:,\quad J^\wedge(K):=\bigcup_{x\in K}J^\wedge(x)
\qquad \text{and} \qquad J(K):=J^\vee(K)\cup J^\wedge(K) \:.
$$ 
Let~$B\Subset M$ be another compact set. Then, it follows by direct inspection that
\begin{equation}\label{causdisconn}
B\subset M\setminus J(K)\qquad\mbox{if and only if}\quad K\subset M\setminus J(B) \:.
\end{equation}
 Condition \eqref{causdisconn} provides us with a sensible notion of {\bf{causally disconnected}} compact sets. 
By construction, this means in particular that for any~$x\in K$ and~$y\in B$ one has~$x\not\prec y$, i.e. there exists a foliation~$(\eta_t)_{t \in \R}$ and a time~$t$ such that (see the left of Figure~\ref{figurestrongCD})
\begin{figure}
\centering
\def\svgwidth{\columnwidth}
\scalebox{0.7}{\hspace{-1.5em}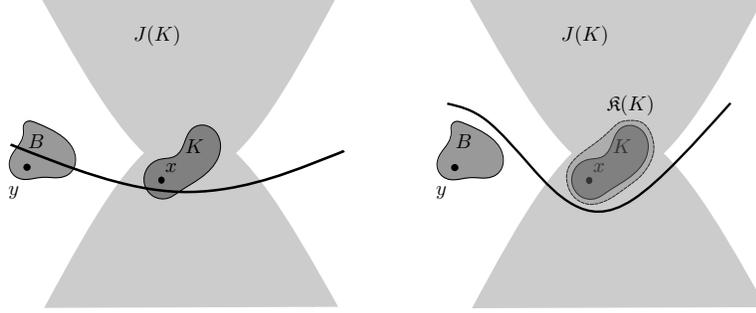}
\caption{Causally and strongly causally disconnected sets.}
\label{figurestrongCD}
\end{figure}
$$
\eta_t(x)=0\quad\mbox{and}\quad\eta_t(y)=1
$$
(the same is true if the roles of~$x$ and~$y$ are interchanged). However, the cutoff function~$\eta_t$ depends on the points~$x$ and~$y$. Moreover, the support of~$\theta_t$ may intersect the sets~$K$ or~$B$ at other points than~$x$ or~$y$ (see again the left of Figure~\ref{figurestrongCD}). In general, it is not clear whether such a separation can be carried out uniformly on~$K$ and~$B$. Moreover, having the canonical commutation relations in mind, we need to take into consideration that the Green's operators vanish in the past/future of the inhomogeneity and are foliation independent.

Referring to Lemma~\ref{lemmacompact} and Remark~\ref{remarkcompact}, we now give the following  stronger definition of causal disconnection (see the right of Figure~\ref{figurestrongCD}).
\begin{Def}\label{strongCD}
Two compact sets~$K,B \Subset M$ are said to be {\bf{strongly causally disconnected}} if they are causally disconnected and if there is~$(\eta_t)_{t \in \R} \in\Theta$ and~$t\in\R$ such that
\[ \eta_{t}|_{\mathfrak{K}(K)}\equiv 0\quad\mbox{and}\quad 	\eta_{t}|_{\mathfrak{K}(B)}\equiv 1 \:, \]
and similarly if the roles of~$K$ and~$B$ are exchanged. This is denoted by~$B\perp K$.
\end{Def}

The next result follows immediately by combining Definition~\ref{strongCD} with Lemma~\ref{lemmacompact} and Remark~\ref{remarkcompact} (with obvious changes for the advanced Green's operator).
\begin{Prp}\label{Propositionspacelikesepvanish}
Let~$\u,\v\in L^2_0(M)$ be such that~$\supp\u\perp \supp\v$. Then
$$
S^\vee\u\,|_{\supp\v}=S^\wedge\u\,|_{\supp\v}=G\,\u\,|_{\supp\v} \equiv 0 \:.
$$
In particular,
\[ \langle \u, S^\vee \v\rangle_{L^2(M)}=	\langle \u, S^\wedge \v\rangle_{L^2(M)}=	\langle \u, G\, \v\rangle_{L^2(M)}=0 \:. \]
\end{Prp}
Proposition~\ref{Propositionspacelikesepvanish} provides us with a precise mathematical formulation of the causality condition anticipated at the end of Section~\ref{sectionfoliation}. One could proceed further, for example by tentatively defining the {\bf{strong causal cone}} generated by a compact~$K$ as
\[ 
J_S(K):=M\setminus \bigcup\{B\Subset M\:|\: B\perp K \} \:. \]
The support of the causal fundamental solution would then be contained in it. Namely,
\begin{Prp}
$
\supp G\,\w\subset J_S(\supp\w)
$
for all~$\w\in L_0^2(M)$.
\end{Prp}
\noindent Similar properties can be shown for the advanced and retarded Green's operators. 

Although these causal properties seem worth being analyzed in more depth,
here we will not proceed in this direction,
because the content of Proposition~\ref{Propositionspacelikesepvanish} suffices for the purposes of this paper.

\subsection{Adjoint Properties of the Green's Operators}\label{sectionadjoint}
The goal of this section consists in proving that the advanced and retarded Green's operators restricted to
a suitable scalar product space are the adjoints of each other.
To this aim, we need additional assumptions, which we now introduce.
We say that a jet in~$L_\loc^2(M)$ is {\em{spatially compact}} if its restriction to any time strip has compact support, no matter what foliation we refer to. This will henceforth be denoted by a lower index~``$\sc$''.
For example,
\[ \Jvary_\sc := \{ \v \in \Jvary \:|\: \text{$\v$ spatially compact} \} \:. \]

\begin{Def}\label{asymptoticallypartitioned}
Spacetime~$M$ is said to be {\bf{asymptotically partitioned}} if there is an exhaustion by
relatively compact sets~$B_n$ and a family of linear operators
$$
P_n:\Jvary_\sc\to \Jvary_0
$$
with the property that, for all~$\v\in \Jvary_\sc$,
\begin{itemize}[leftmargin = 2.5em]
\item[\rm{(i)}] $P_n\v\equiv \v$ in~$B_{n}$ and~$P_n\v\equiv 0$ outside~$B_{n+1}$.\\[-0.8em]
\item[\rm{(ii)}] $\|P_n\v\|_x\le \|\v\|_x$ and~$P_n\v\to \v$ in~$L^2_\loc(M)$.
\end{itemize}
\end{Def}
\noindent As a direct consequence of point (ii), we see that
\[ 
\Jvary_0\mbox{ is dense in } \Jvary_\sc\mbox{ with respect to the topology of~$L^2_\loc(M)$ \:.} \]
In the example when~$\Jvary$ is chosen as the space of all smooth functions on a manifold,
the mapping~$P_n$ can be introduced as the multiplication with a bump function. In our setting, however, the
set~$\Jvary$ will in general not be closed under multiplication by smooth functions.
This is why we need to impose this property as a separate condition.

We now introduce the following jet space,
\beq \label{JGr}
(\Jvary_0)^*:=\left\{\v\in L^2_{0}(M)\:|\: S^\vee\v,\,S^\wedge\v \in \Jvary_\sc\right\}.
\eeq
On this space, the Green's operators are indeed adjoints of each other
with respect to the standard $L^2$-scalar product:

\begin{Prp} \label{propadjoingS}
If spacetime~$M$ is asymptotically partitioned, then
$$
\langle S^\wedge \u,\v\rangle_{L^2(M)}=\langle  \u,S^\vee\v\rangle_{L^2(M)}\qquad\mbox{for all~$\u,\v\in (\Jvary_0)^*$}\,.
$$
\end{Prp}
\begin{proof}
Let~$\u,\v \in (\Jvary_0)^*$. By definition, these jets
are mapped by the Green's operators to elements of~$\Jvary_\sc$.
We choose~$N$ so large that~$\supp\v\subset B_n$ for all~$n\ge N$.
Then, using the definition of global weak solution, we have
\[ \langle S^\wedge\u,\v\rangle_{L^2(M)} =\langle P_n( S^\wedge\u),\v\rangle_{L^2(M)}=\langle \Delta(P_n\, S^\wedge\u),S^\vee\v\rangle_{L^2(M)} \:. \]
Using that the jet~$P_n\, S^\wedge\u$ is compactly supported and that the Lagrangian has compact range, 
we can apply the Green's formula (with vanishing boundary terms) to obtain
\[ \langle S^\wedge\u,\v\rangle_{L^2(M)}
= \langle P_n( S^\wedge\u),\Delta(S^\vee\v)\rangle_{L^2(M)} \:. \]

Taking the limit~$n\to \infty$, using Definition~\ref{asymptoticallypartitioned}~(ii) (choosing a subsequence if necessary) and applying Lebesgue's dominated convergence theorem, we obtain
$$
\langle S^\wedge\u,\v\rangle_{L^2(M)}=\langle S^\wedge\u,\Delta(S^\vee\v)\rangle_{L^2(M)}=\langle \Delta(S^\wedge\u),S^\vee\v\rangle_{L^2(M)},
$$
where in the last step we again used the Green's formula. Again applying Lebesgue's dominated
convergence theorem gives
\begin{equation*}
\begin{split}
\langle \Delta(S^\wedge\u),S^\vee\v\rangle_{L^2(M)} & = \lim_{n\to\infty }\langle \Delta(S^\wedge\u),P_n (S^\vee\v)\rangle_{L^2(M)} \\
&=\lim_{n\to\infty }\langle \u,P_n (S^\vee\v)\rangle_{L^2(M)}
=\langle \u,S^\vee\v\rangle_{L^2(M)} \:.
\end{split}
\end{equation*}
Putting all together, we get the claim.
\end{proof}

We next analyze the kernel of the Green’s operators. Note that the inhomogeneity~$\w$ in the global weak equation~\eqref{weakequation} can be
changed arbitrarily by jets in the orthogonal complement of the test space. More precisely, the right side
of~\eqref{weakequation}
vanishes identically for any~$\w\in (\Jvary_0)^\perp\cap L_0^2(M)$, where the symbol~$\perp$ refers to the standard scalar product of~$L^2(M)$. This suggests that also the Green’s
operators should vanish for such jets. This is indeed the case, as shown
in the next lemma. 
\begin{Lemma}\label{lemmakernel}
The Green's operators vanish on the orthogonal complement of~$\Jvary_0$, i.e.
\[ (\Jvary_0)^\perp\cap L_{0}^2(M)\subset \ker S^\wedge\cap \ker S^\vee \:. \]
\end{Lemma}
\begin{proof} We only prove the inclusion in~$\ker S^\wedge$; the argument for~$\ker S^\vee$ is analogous.
Let~$\w \in L^2_0(M)$ be orthogonal to all jets in~$\Jvary_0$.
In view of the definition of~$S^\wedge \w$ as a limit (see Theorem \eqref{global}~(i)),
it clearly suffices to show that the weak solution~$\v(L_s^t,\w)$ vanishes in
sufficiently large time strips. To this end, we return to the existence proof of Theorem~\ref{global}.
We choose~$L_s^t$ so large that~$\supp\w\stackrel{\circ}{\subset}L_s^t$.
Then for any~$\u \in \overline{\J}^t$,
\[ \la \w, \u \ra_{L^2(L_s^t)} = \la \w , \u \ra_{L^2(M)} = 0 \:. \]
Therefore, the linear functional~$\la \w , .\, \ra_{L^2(L_s^t)}$ vanishes on a dense subset
of the Hilbert space~$\overline{\mathcal{H}}(L_s^t)$, implying that it is represented by the zero vector~$0=V \in \overline{\mathcal{H}}(L_s^t)$.
As a consequence, also the weak solution~$\v(L_s^t,\w)=\Delta V$ vanishes.
\end{proof}
\begin{Remark}
The reason why the set~$(\Jvary_0)^\perp\cap L_{0}^2(M)$ is non-trivial in general lies in the fact that no denseness assumption of~$\Jvary$ in~$\Jtest$ has been made.
\end{Remark}

As a consequence of the last lemma, the bilinear form in Proposition~\ref{propadjoingS} is independent of the representatives and extends to a bilinear form on the quotient space
\[ \J_0^* := \bigslant{ (\Jvary_0)^*}{(\Jvary_0)^\perp} \]
To see this, let~$\u\sim \u'$ and~$\v\sim \v'$ be different representatives of two arbitrary equivalence classes of~$\J_0^*$. Thanks to Lemma~\ref{lemmakernel} we know that~$S^\wedge\u=S^\wedge\u'$ and~$S^\vee\v=S^\vee\v'$. Therefore, using Proposition~\ref{propadjoingS},
$$
\langle S^\wedge\u,\v\rangle_{L^2(M)}=\langle S^\wedge\u',\v\rangle_{L^2(M)}=\langle \u',S^\vee\v\rangle_{L^2(M)}=\langle \u',S^\vee\v'\rangle_{L^2(M)}=\langle S^\wedge\u',\v'\rangle_{L^2(M)}\,.
$$ 
We conclude that the product~$\langle S^\wedge \u,\v\rangle_{L^2(M)}$ is well-defined on~$\J_0^*$,
as it does not depend on the chosen representatives. A similar argument applies to the advanced Green's operator.
This result will be the starting point of next section.

\subsection{The Symplectic Form and the Time Slice Property}\label{sectionSFandTSA}
Having the canonical commutation relations in mind, we now introduce the {\bf{pre-symplectic form}}~$G$ by
\begin{equation}\label{preform}
G \::\: \J_0^*\times \J_0^*\ni ([\u],[\v])\mapsto \langle \u, G \v\rangle_{L^2(M)}\in \R\:.
\end{equation}
Using~\eqref{Gdef} together with the adjoint properties of the Green's operators in Proposition~\ref{propadjoingS},
one sees that this functional is antisymmetric. Moreover, from Proposition~\ref{Propositionspacelikesepvanish}
one concludes that~$G$ vanishes on jets whose supports are strongly causally disconnected.
Thus
\begin{itemize}[leftmargin=2.5em]
\item[{\rm (i)}] $G([\u],[\v])=0$ if~$\supp\u\perp\supp\v$\\[-1em]
\item[{\rm (ii)}] $G([\u],[\v])=-G([\v],[\u])$.
\end{itemize}
The reason why~$G$ is only a \textit{pre}-symplectic form is that it may be degenerate in the sense that
its kernel~$N$ defined by
\[ 
N:=\{[\u] \in \J_0^*\:|\: G([\u],\,.\,) = 0\} \:. \]
may be non-trivial.  
Therefore, in order to obtain a \textit{non-degenerate} symplectic form, we need to divide out
the subspace~$N$.

\begin{Def}\label{classobs}
The space of {\bf{classical fields}} is the pair~$(\mathcal{E}(M),G)$, where
$$
\mathcal{E}(M):= \bigslant{\J_0^*}{\big\{ [\u]\:|\: \u\in N \big\}} \:,
$$
and~$G:\mathcal{E}(M)\times\mathcal{E}(M)\rightarrow\C$ is the {\bf{symplectic form}} obtained from \eqref{preform}.
\end{Def}

In the setting of hyperbolic PDEs in globally hyperbolic spacetimes, the {\em{time slice property}}
states that, given an open neighborhood of a Cauchy surface, every vector in the symplectic
space~$\mathcal{E}(M)$ has a representative supported in this neighborhood.
We conclude this section by analyzing in which sense and under which assumptions
the time slice property holds for linearized fields in the setting of causal variational principles.

\begin{Def}\label{cutoff}
Let~$L_s^t$ be a time strip in a given foliation~$(\eta_t)_{t \in \R}$. A {\bf{cutoff operator}} in~$L_s^t$ is a linear operator~$\check{\pi}: \Jvary_\sc\to \Jvary_\sc$ such that
$$
\eta_{s} \,(1-\check{\pi}) = 0 = (1-\eta_{t}) \,\check{\pi} \:.
$$
\end{Def}
Let~$\v\in (\Jvary_0)^*$. By definition, $S^\wedge \v$ and~$S^\vee \u$ belong to~$\Jvary_\sc$. Moreover, from the support properties of the Green's operators we know that these two weak solutions are past and future compact, respectively. From Definition~\ref{cutoff}, it follows that
\begin{equation}\label{cutsupport}
\check{\pi}(S^\wedge \v)\in \Jvary_0\ \mbox{ and }\ (1-\check{\pi})(S^\vee \v)\in \Jvary_0.
\end{equation}
The proof of our time slice property (see Theorem~\ref{timesliceaxiom} below)
will require a few additional assumptions on the space~$\Jvary$. 
Inspired by differential operators and smooth sections on globally hyperbolic Lorentzian manifolds we impose the following conditions:
\beq \label{invarianceJvary}
\begin{split}
{\rm{(a)}} &\quad \Jvary\cap (\Jvary_0)^\perp=\{0\} \\
{\rm{(b)}} &\quad \Delta(\Jvary)\subset\Jvary \qquad \text{and hence} \qquad \Delta(\Jvary_\sc)\subset \Jvary_\sc \:,
\qquad\qquad\qquad\qquad
\end{split}
\eeq
where the last inclusion follows from the fact that the Lagrangian has compact range. 
As a direct consequence, one has the following result.
\begin{Lemma}\label{lemmainverseDelta}
For every~$\v\in (\Jvary_0)^*$ there exist~$E\in (\Jvary_0)^\perp\cap L_0^2(M)$ such that
\begin{equation*}
\begin{split}
&\Delta (S^\wedge\v)=-\v+E=\Delta (S^\vee\v)\qquad\text{and}\qquad
\Delta (G \v)=0.
\end{split}
\end{equation*}
In particular, $E=0$ whenever~$\v\in\Jvary$.
\end{Lemma}
\begin{proof}
Note that~$S^\wedge\u\in\Jvary$ for every~$\u\in (\Jvary_0)^*$, and hence~$\Delta S^\wedge\u$ is well-defined. In particular, this jet belongs to~$\Jvary\subset L^2_\loc(M)$.

From the definition of global weak solution and the Green's formula, we get 
$$
\langle \Delta S^\wedge\v,\u\rangle_{L^2(M)}=\langle  S^\wedge\v,\Delta\u\rangle_{L^2(M)}=\langle  \v,\u\rangle_{L^2(M)}\quad\mbox{for every~$\u\in\Jvary_0$}.
$$
Thus, $\Delta S^\wedge\v+\v=:E^\wedge\in (\Jvary_0)^\perp$.
Next, from the assumption of compact range and the support property of the retarded Green's operator
we infer that~$E^\wedge$ must vanish sufficiently far in the past of~$\v$. A similar argument applies to the advanced Green's operator, yielding an analogous jet~$E^\vee\in (\Jvary_0)^\perp$ which vanishes sufficiently far in the future of~$\v$. Now, note that~$G \v\in \Jvary_\sc$ and hence, using~(b) in~\eqref{invarianceJvary} and the definition of weak global solution,
  $$
  \Jvary_\sc\ni \Delta(G \v)=E^\vee-E^\wedge\quad\mbox{and}\quad \langle \Delta G \v,\u\rangle_{L^2(M)}=0\quad\mbox{for all $\u\in\Jvary_0$},
  $$
  where the last identity follows again from the Green's formula. At this point, using~(a)
  in~\eqref{invarianceJvary}, we obtain~$E^\vee-E^\wedge\in\Jvary_\sc\cap(\Jvary_0)^\perp=\{0\}$ and hence~$E^\vee=E^\wedge=:E.$ Because~$E^\wedge$ and~$E^\vee$ vanishes in the past and the future of~$\v$, respectively, we conclude that~$E$ has compact support. The last statement is clear from \eqref{invarianceJvary}.
\end{proof}

We are now ready to prove the time slice property. As an additional technical condition, we need
to assume that the causal Green's operators map~$\Jvary_0$ to~$\Jvary_\sc$,
\[ S^\vee,\,S^\wedge \::\: \Jvary_0 \rightarrow \Jvary_\sc \:. \]
In view of the definition of~$(\Jvary_0)^*$ in~\eqref{JGr}, this condition can be
written in the shorter form
\beq \label{condc}
{\rm{(c)}} \quad \Jvary_0\subset (\Jvary_0)^* \:. \qquad\qquad\qquad\qquad\qquad\qquad\qquad\qquad
\qquad\qquad\qquad\qquad
\eeq

We say that $L_{s_0}^{t_0}$ is \textit{$\Delta$-contained} in $L^{t_1}_{s_1}$ if the width of $L_{s_1}^{t_1}\setminus L_{s_0}^{t_0}$ is larger than the range of the operator $\Delta$. More precisely, we demand that, for every $\u\in \Jvary$,
\beq \label{deltacontained}
\eta_{s_0}\u\equiv 0 \;\Longrightarrow\; \eta_{s_1}\Delta\u\equiv 0\qquad\mbox{and}\qquad(1-\eta_{t_0})\u\equiv 0 \;\Longrightarrow\; (1-\eta_{t_1})\Delta\u\equiv 0 \:.
\eeq
\begin{Thm} {\bf{(Time Slice Property)}} \label{timesliceaxiom}
Assume that the conditions~{\rm{(a)--(c)}} in~\eqref{invarianceJvary} and~\eqref{condc} hold.
Let~$L_0,L_1$ be time strips with respect to the same foliation such that
\begin{itemize}[leftmargin=2em]
\item[{\rm (i)}] $L_0$ admits a cutoff operator~$\check{\pi}$,\\[-1em]
\item[{\rm (ii)}] $L_0$ is $\Delta$-contained in $L_1$ (see~\eqref{deltacontained}). \\[-1.2em]
\end{itemize}
Then for every~$[\v]\in \mathcal{E}(M)$ there exists~$\v'\sim\v$ with~$\supp \v'\stackrel{\circ}{\subset} L_1$.
\end{Thm}
\begin{proof}
The following is an adaptation of the proof of \cite[Chapter 3, Theorem 3.3.1]{brunettibook}. 
Let~$\v\in (\Jvary_0)^*$. 
From Lemma~\ref{lemmainverseDelta} we know that~$\Delta(G \v)=0$. Hence~$\Delta(\check{\pi}(G \v))+\Delta((1-\check{\pi})(G\v))=0$.
We introduce the jet~$\u$ by (cf.~\eqref{invarianceJvary})
$$
\u:=\Delta((1-\check{\pi})(G\v))=-\Delta(\check{\pi}(G\v))\in\Jvary_\sc \:.
$$
Using that~$\Delta$ has compact range, we infer that the left-hand side vanishes sufficiently far in the past of~$L_0$, while the right-hand side vanishes sufficiently far in the future of~$L_0$.
Moreover, it follows from~(b) in~\eqref{invarianceJvary} that~$\u\in \Jvary_\sc$.
We thus conclude that~$\u\in\Jvary_0$.
More precisely, from the definition of cutoff operator and assumption~(ii) we conclude that~$\u$ is properly supported within~$L_1$, i.e.
$
\supp\u\stackrel{\circ}{\subset} L_1.
$
Moreover, by assumption, $\u\in (\Jvary_0)^*$.
It remains to show that~$\u\sim \v$. 
From Lemma~\ref{lemmainverseDelta}  and the definition of~$\u$ we obtain
\begin{equation*}
\begin{split}
\quad\u-\v&=-\Delta(\check{\pi}(G\v))+\Delta(S^\vee\v) + E \\
&=-\Delta(\check{\pi}(S^\vee\v))+\Delta(\check{\pi}(S^\wedge\v))+\Delta(\check{\pi}(S^\vee\v))+\Delta((1-\check{\pi})(S^\vee\v)) + E \\
&=\Delta\big(\check{\pi}(S^\wedge\v)+(1-\check{\pi})(S^\vee\v)\big) + E =\Delta\w + E
\end{split}
\end{equation*}
with~$E\in (\Jvary_0)^\perp\cap L_0^2(M)$,
where~$\w:=\check{\pi}(S^\wedge\v)+(1-\check{\pi})(S^\vee\v)\in \Jvary_0$, as follows from \eqref{cutsupport}. Let~$\u_0$ be the representative of an equivalence class of~$\J_0^*$. Then, by definition of global weak solutions, we have 
\[ G([\Delta\w],[\u_0])=\langle \Delta\w,G\u_0\rangle_{L^2(M)}=0 \:. \]
In other words, $\Delta\w\in N$. 
The claim follows from the definition of~$\mathcal{E}(M)$.
\end{proof}

\section{The Algebra of Fields} \label{secalgebra}
In this section we give an application of the results obtained in the previous sections. More precisely, motivated by the so-called algebraic approach to quantum field theory~\cite{brunettibook}, we show how to associate with the linearized fields individuated by a causal variational principle a distinguished algebra of fields.
This should play the r\^{o}le of the building block for the quantization of the underlying system. Therefore, it must encode information both on the dynamics and on the canonical commutation relations. 

\subsection{Construction of the Algebra}
The starting point of our construction is~$\mathcal{E}(M)$, the collection of classical fields as per Definition~\ref{classobs}. We complexify
\[ \mathcal{E}(M)^\C := \mathcal{E}(M) \otimes \C \]
and form the corresponding universal tensor algebra 
\[ \mathcal{T}(M):=\bigoplus\limits_{n=0}^\infty \big( \mathcal{E}(M)^\C \big)^{\otimes n} \:, \]
where we set~$(\mathcal{E}(M)^\C)^{\otimes 0}:=\mathbb{C}$. This is a $*$-algebra if endowed with a $*$-structure out of the natural extension of complex conjugation to the tensor product.

We observe that we have already encoded the information on the underlying dynamics, because
$\mathcal{E}(M)$ is built in terms of a quotient between~$\J_0^*$ and the subspace~$N$ which
includes the kernel of the causal fundamental solution~$G:=S^\wedge-S^\vee$. In order to codify the
counterpart in this setting of the canonical commutation relation we proceed as follows:

\begin{Def}\label{CCRalgebra}
The {\bf algebra of fields} is the $*$-algebra built as the quotient
\[ \mathcal{A}(M):= \bigslant{\mathcal{T}(M)}{\mathcal{I}(M)} \:, \]
where~$\mathcal{I}(M)$ is the $*$-ideal generated by elements of the form
\[ [\u]\otimes[\v]-[\v]\otimes[\u]-i G \big( [\u],[\v] \big) \:\mathbb{I} \:, \]
where~$\mathbb{I}$ is the identity of~$\mathcal{T}(M)$ and~$G$ as defined in~\eqref{preform}. 
\end{Def} \noindent
The algebra of fields is also referred to as the {\em{algebra of observables}}.

\subsection{Properties of the Algebra} \label{secalgebraprop}
Definition~\ref{CCRalgebra} is clearly reminiscent of the usual construction of the algebra of fields of a bosonic real scalar field theory, see~\cite{brunettibook}. 
We now show that many, but a-priori not all, of the standard properties hold in this setting. 
To this end, we show that the properties of the classical dynamics of the linearized fields
translate to corresponding properties of the algebra.

\begin{Prp} {\bf{(causality)}} Let~$\u, \v \in (\Jvary_0)^*$ whose supports are strongly causally disconnected
(see Definition~\ref{strongCD}),
\[ \supp \u \perp \supp \v \:. \]
Then the corresponding field operators commute,
\[ \big[ [\u], [\v] \big] = 0 \:. \]
\end{Prp}
\Proof
 Since~$\mathcal{A}(M)$ and in turn~$\mathcal{T}(M)$ are generated by~$\mathcal{E}(M)^\C$, 
 it suffices to work at the level of generators. In particular, for any~$[\u],[\v]\in\mathcal{E}(M)^\C$,
the commutator is given by~$\left[[\u],[\v]\right]=iG([\u],[\v])\mathbb{I}$.
 Yet, in view of~\eqref{preform}, $G([\u],[\v])=0$ if~$[\u]\perp[\v]$.
\QED

\begin{Prp} {\bf{(time slice property)}} Referring to the assumptions of Theorem~\ref{timesliceaxiom},
we assume that~$\Jvary_0\subset (\Jvary_0)^*$ and let~$L_0,L_1$ be time strips
with respect to the same foliation having the properties~{\rm{(i)}} and~{\rm{(ii)}} on page~\pageref{timesliceaxiom}.
Then the algebra of fields is generated by jets which are properly contained in~$L_1$
(see Definition~\ref{defproperlycontained}),
\[ \mathcal{A}(M) = \Big\la \Big\{ [\v] \text{ with } \v \in (\Jvary_0)^* \text{ and }\supp \v \stackrel{\circ}{\subset} L_1 \Big\} \Big\ra\:. \]
In other words, there exists a $*$-isomorphism between~$\mathcal{A}(L_1)$ and~$\mathcal{A}(M)$.
\end{Prp}
\Proof
This is a direct consequence of Theorem~\ref{timesliceaxiom} and of the properties of the symplectic form~$G$ defined by~\eqref{preform}.
\QED

We finally remark that symmetries of the classical system as expressed in terms of groups 
of symplectomorphisms acting on the linearized fields extend to $*$-auto\-mor\-phisms of the algebra
of fields. These constructions are straightforward, and we omit them here
and refer instead to~\cite[Chapters~3 and~5]{brunettibook}.

\subsection{Construction of Distinguished Quasi-Free States} \label{secrepresent}
Following the standard rationale at the heart of the algebraic approach to quantum field theory, in addition to individuating an algebra of fields, the quantization procedure is complete only after one selects an algebraic state, namely in the case in hand a positive and normalized linear functional 
\begin{equation}\label{algstate}
\omega \::\:
\mathcal{A}(M)\to\mathbb{C}\;\;\textrm{such that}\;\;\omega(\mathbb{I})=1\;\;\;\textrm{and}\;\;\;\omega(a^*a)\geq 0\;\;\;\forall a\in\mathcal{A}(M)\:.
\end{equation}

The renown GNS theorem~\cite{brunettibook} entails that, to any given pair~$(\mathcal{A}(M),\omega)$, one can associate a unique (up to $*$-isomorphisms) triple~$(\mathcal{D}_\omega,\pi_\omega,\Omega_\omega)$. Here~$\mathcal{D}_\omega$ is a dense subset of a Hilbert space, $\pi_\omega:\mathcal{A}(M)\to\mathcal{L}(\mathcal{D}_\omega)$ a $*$-representation of the algebra in terms of linear operators, while~$\Omega_\omega$ is a unit norm, cyclic vector such that, for any~$a\in\mathcal{A}(M)$, it holds~$\omega(a)=(\Omega_\omega,\pi_\omega(a)\Omega_\omega)$.

Among the plethora of algebraic states a distinguished r\^{o}le is played by the quasi-free/Gaussian ones. In order to characterize them, let us observe that a generic element~$a\in\mathcal{A}(M)^\C$ can be decomposed as 
\[ a=a_0\, \mathbb{I}+a_1\,[\u]+a_2\, [\u_1]\otimes[\u_2]+\cdots \:, \]
where~$a_0,a_1,\dots\in\mathbb{C}$ while~$[\u],[\u_1],\dots\in\mathcal{E}(M)^\C$. Hence, given any algebraic state~$\omega$, linearity entails that one can associate with it the so-called {\em n-point functions}:
\[ \omega_n([\u_1],\dots,[\u_n])=\omega([\u_1]\otimes\dots\otimes[\u_n]),\quad[\u_1]\dots[\u_n]\in\mathcal{E}(M)^\C \:. \]
The next definition follows~\cite[Def.~5.2.22 in Chapter~5]{brunettibook} and~\cite{kay-wald}.

\begin{Def} \label{defquasifree}
Let~$\omega:\mathcal{A}(M)\to\mathbb{C}$ be an algebraic state as per~\eqref{algstate}. We call it {\bf quasi-free/Gaussian} if the associated odd $n$-point functions are all vanishing, while the even ones are such that 
$$\omega_{2n}([\u_1]\dots[\u_{2n}])=\sum\limits_{\sigma\in\mathcal{P}}\omega_2([\u_{\sigma(1)}],[\u_{\sigma(2)}])\dots\omega_2([\u_{\sigma(2n-1)}],[\u_{\sigma(2n)}]),$$
where~$\mathcal{P}$ denotes all possible permutations of the set~$\{1,\dots,2n\}$ into a collection of elements~$\{\sigma(1),\dots,\sigma(2n)\}$ such that~$\sigma(2k-1)<\sigma(2k)$ for all~$k=1,\dots,n$.
\end{Def}

The net advantage of working with Gaussian states is that the associated GNS representation yields a Hilbert space of Fock type, see~\cite{kay-wald}, allowing thus a close connection with the applications to high energy physics. Yet this does not suffice as one can readily infer that several pathological situations can still occur. While in standard quantum field theory this is resolved by introducing the notion of Hadamard states, a counterpart of such concept in this setting is still beyond our grasp. Nonetheless we can outline the construction of a distinguished algebraic quasi-free state making use of a complex structure on the linearized fields
as first obtained in~\cite[Section~6.3]{fockbosonic}. This complex structure gives rise to
a splitting of the complexified solutions space into two subspaces, referred to as the
holomorphic and anti-holomorphic components. Mimicking the procedure for the frequency splitting
in linear quantum field theory will give the desired quasi-free state.

Before beginning, we recall that in Section~\ref{secosi} we introduce two bilinear forms on the linearized fields:
the surface layer inner product~\eqref{jipdef} and the symplectic form~\eqref{sympdef}.
The symplectic form is conserved in the sense that, if~$\u$ and~$\v$ are linearized solutions,
then~$\sigma^t(\u,\v)$ is time independent. This conservation law is expressed by the
Green's formula in Lemma~\ref{green}. The surface layer inner product, on the other hand,
in general does {\em{not}} satisfy a general conservation law (for details see~\cite{osi}).
Moreover, in general it depends on the choice of the foliation. Nevertheless, for a given foliation
and at any given time, the surface layer inner product can be used to endow the
linearized solutions with a complex structure. This also gives rise to a corresponding quasi-free state,
as we now work out.

We next introduce the complex structure following the procedure in~\cite[Section~6.3]{fockbosonic}.
We first extend the surface layer inner product~\eqref{jipdef} to jets with spatially compact support,
\[ (.,.)^t \::\: \Jvary_\sc \times \Jvary_\sc \rightarrow \R \:. \]
Similar as in the hyperbolicity conditions in Definition~\ref{globfoliate} we assume that,
the restriction of this bilinear form to solutions of the linearized field equations
is positive semi-definite.
Dividing out the null space and forming the completion, we obtain
a real Hilbert space denoted by~$\h^\R$.
Next, we assume that~$\sigma$ is a bounded bilinear functional on this Hilbert space.
Then we can represent it relative to the scalar product by
\[ 
\sigma^t(u, v) = (u,\, \mathscr{T}\, v)^t \:, \]
where~$\mathscr{T}$ is a uniquely determined bounded operator on~$\h^\R$.
Since the symplectic form is anti-symmetric and the scalar product is symmetric, it is obvious that
\[ \mathscr{T}^* = -\mathscr{T} \]
(where the adjoint is taken with respect to the scalar product~$(.,.)$).
Finally, we assume that~$\mathscr{T}$ is invertible.
Then setting
\beq \label{Idef}
J := -(-\mathscr{T}^2)^{-\frac{1}{2}}\: \mathscr{T}
\eeq
defines a complex structure on the real Hilbert space~$\h^\R$.

We next complexify the vector space~$\Jlin_\sc$ and denote
its complexification by~$\J^\C$. We also extend~$J$ to a complex-linear operator
on~$\J^\C$. The fact that~$J^*=-J$ and~$J^2=-\1$ implies that~$J$ has
the eigenvalues~$\pm i$. Consequently, $\J^\C$ splits into a direct sum of the
corresponding eigenspaces, which we refer to as the {\em{holomorphic}} and 
{\em{anti-holomorphic subspaces}}, i.e.\
\[ 
\J^\C = \J^\hol \oplus \J^\ah \qquad \text{with} \qquad
\J^\hol := \chi^\hol \:\J^\C \:,\;\;\;
\J^\ah := \chi^\ah \:\J^\C \:, \]
where we set
\[ 
\chi^\hol = \frac{1}{2}\: (\1 - i J) \qquad \text{and} \qquad
\chi^\ah = \frac{1}{2}\: (\1 + i J) \:. \]
We also complexify the inner product~$(.,.)$ and the symplectic form to {\em{sesquilinear}}
forms on~$\J^\C$ (i.e.\ anti-linear in the first and linear in the second argument).
Moreover, we introduce a positive semi-definite inner product~$(.|.)$ by
\[ 
(.|.) = (\,.\,,(-\mathscr{T}^2)^\frac{1}{2}\, .\,) = \sigma( \,.\,, J \,.\, ) \::\: \J^\C \times \J^\C \rightarrow \C \:. \]
This positive semi-definite inner product gives rise to a Hilbert space structure.
In order to work out the similarities and differences to quantum theory, it is best to
form the Hilbert space as the completion of the holomorphic subspace, i.e.\
\[ 
\h := \overline{\J^\hol}^{(.|.)} \:. \]
We denote the induced scalar product on~$\h$ by~$\la.|. \ra$. Then~$(\h, \la.|.\ra)$
is a complex Hilbert space.

In order to define a quasi-free state, according to Definition~\ref{defquasifree} it suffices
to specify the two-point function~$\omega_2$. We define it by
\beq \label{w2def}
\omega_2 \big([\u], [\v] \big) := i \sigma^t \big( G \u, \chi^\hol G \v \big) \:.
\eeq

In order to show compatibility with the canonical commutation relations, we need to
get a connection between~$G$ and the symplectic form~$\sigma^t$ defined by~\eqref{sympdef}.
We first note that for the latter bilinear form to be well-defined, it suffices
to assume that the jets have spatially compact support. We thus obtain a bilinear form
\[ \sigma^t(.,.) \::\: \Jvary_\sc \times \Jvary_\sc \rightarrow \R \:. \]
\begin{Prp} \label{prpsigma}
For any jets~$\u, \v \in (\Jvary_0)^*$,
\[ G \big( [\u], [\v] \big) = \sigma^t \big( G \u, G \v \big) \:, \]
where the last surface layer integral can be computed in any foliation~$(\eta_t)_{t \in \R}$ at any time~$t$.
\end{Prp}
\Proof We follow the method in~\cite[proof of Proposition~5.10]{linhyp}.
We let~$\u, \v \in (\Jvary_0)^*$ be arbitrary representatives.
Since~$G \u, G \v \in \Jvary_\sc$ are linearized solutions, the Green's formula~\eqref{eqgreen}
implies that the symplectic form~$\sigma^t$ is time independent. Choosing~$t$ sufficiently large,
the advanced Green's operators applied to~$\u$ and~$\v$ vanish in the surface layer at time~$t$. Therefore,
\[ \sigma^t \big( G \u, G \v \big) = \sigma^t \big( S^\wedge \u, S^\wedge \v \big) \:. \]
Now we apply again the Green's formula~\eqref{eqgreen} in a time strip~$L=L^t_s$, where we choose~$s$ so large
that the jets~$S^\wedge \u$ and~$S^\wedge \v$ vanish in the surface layer at time~$s$. We thus obtain
\begin{align*}
\sigma^{t} &\big( S^\wedge\u, S^\wedge\v \big) = \la S^\wedge\u, \Delta S^\wedge \v\ra_{L^2(L)} - \la \Delta S^\wedge\u, S^\wedge \v\ra_{L^2(L)} \\
&= -\la S^\wedge \u, \v\ra_{L^2(L)} + \la \u, S^\wedge \v \ra_{L^2(L)}
= \la \u, \big( S^\wedge -S^\vee \big) \v\ra_{L^2(L)} = \la \u, G \v\ra_{L^2(L)} \:.
\end{align*}
In the last integral, we can extend the integration range to all of~$M$, giving~$G([\u], [\v])$.
\QED

\begin{Prp}
The two-point function~\eqref{w2def} defines a quasi-free quantum state.
\end{Prp}
\Proof Our task is to verify the positivity statement in~\eqref{algstate} and the
compatibility with the canonical commutation relations. Before beginning, we point out
that the operator~$\mathscr{T}$ maps the real Hilbert space~$\h$ to itself.
As a consequence, the same is true for the operator~$J$ in~\eqref{Idef}.
Therefore, for~$\u, \v \in (\Jvary_0)^*$ and~$g$, we can decompose~\eqref{w2def} into its real and
imaginary parts,
\begin{align*}
\omega_2 \big([\u], [\v] \big)
&=i \sigma^t \big( G \u, \chi^\hol\, G \v \big) 
= \frac{i}{2} \sigma^t \big( G \u, (\1 - i J) G \v \big) \\
&= -\frac{1}{2}\:\sigma^t \big( G \u, (-\mathscr{T}^2)^{-\frac{1}{2}}\: \mathscr{T} \,G \v \big) + \frac{i}{2}\: \sigma^t \big( G \u, G \v \big) \\
&= \frac{1}{2}\:\big( G \u \,\big|\, (-\mathscr{T}^2)^{\frac{1}{2}}\,G \v \big) + \frac{i}{2}\: \sigma^t \big( G \u, G \v \big) \:.
\end{align*}
In particular, we conclude that the real part is positive semi-definite, and that the imaginary part satisfies the relation
\[ \im \omega_2 \big([\u], [\v] \big) = \frac{1}{2}\:\sigma^t \big( G \u, G \v \big) = \frac{1}{2}\:G \big( [\u], [\v] \big) \:, \]
where in the last step we applied Proposition~\ref{prpsigma}.
Now the result follows from Proposition~5.2.23~(b) in the
textbook~\cite{brunettibook}.
\QED

We conclude this paper with a discussion of the significance and uniqueness properties of the
constructed quasi-free state. As already mentioned at the beginning of this section, the
surface layer inner product~\eqref{jipdef} in general does depend on~$t$ and the choice of the
foliation. However, for non-interacting systems like the so-called {\em{linear systems
in Minkowski space}} as introduced in~\cite[Section~6.1]{fockbosonic},
the surface layer inner product is indeed conserved for linearized solutions.
As a consequence, the quasi-free state defined above becomes independent of time and
of the choice of foliation. This result is compatible with the fact that in Minkowski space,
there is a unique vacuum state obtained by frequency splitting.
Indeed, the explicit analysis in~\cite{action} shows that for Dirac systems in Minkowski space,
the complex structure defined by~\eqref{Idef} indeed gives back frequency splitting
(in the sense that holomorphic jets have positive frequency, and anti-holomorphic jets have
negative frequency).
In this setting, the quasi-free quantum field theories constructed in the present paper
can be regarded as quantum fields involving an ultraviolet regularization for the degrees
of freedom as described by the corresponding causal variational principle.

In nonlinear interacting systems or in systems not defined in Minkowski space,
we do not get one distinguished state, but instead a whole family of states
parametrized by time and the chosen foliation.
In analogy to the class of Hadamard states, it can be regarded as a family
of physically sensible states. More precisely, taking into account the ultraviolet regularization,
our quasi-free states should be of the {\em{regularized Hadamard form}} as introduced
and analyzed in~\cite{reghadamard}. However, making this connection precise, one would
have to consider families of causal variational principles parametrized by
the regularization length~$\varepsilon>0$ and analyze the asymptotic behavior of the
linearized solutions for small~$\varepsilon$. This analysis seems an interesting project for the future.

\Thanks{{{\em{Acknowledgments:}} We are grateful to the referee for valuable comments.
C.D.\ is grateful to the Vielberth Foundation, Regensburg,
for generous support.

\providecommand{\bysame}{\leavevmode\hbox to3em{\hrulefill}\thinspace}
\providecommand{\MR}{\relax\ifhmode\unskip\space\fi MR }
\providecommand{\MRhref}[2]{%
  \href{http://www.ams.org/mathscinet-getitem?mr=#1}{#2}
}
\providecommand{\href}[2]{#2}

\end{document}